\documentclass[mmnp]{edpsmath}
\usepackage{graphicx}
\usepackage{cite}
\usepackage[T1]{fontenc}
\usepackage{amsmath} 
\usepackage{amssymb} 
\usepackage{amsthm}
\newtheorem{theorem}{Theorem}
\newtheorem{lemma}[theorem]{Lemma}
\usepackage{xcolor}
\begin{document}
\title{Conditional Success of Adaptive Therapy: The Role of Treatment-Holiday Thresholds and Non-Existence of Optimal Strategies Revealed by Mathematical Modelling and Optimal Control}
\author{Lanfei Sun $^\dagger$}\address{School of Mathematics Statistics and Mechanics, Beijing University of Technology, Beijing 100124, China}
\author{Haifeng Zhang $^\dagger$}\address{School of Mathematical Sciences, Jiangsu University, Zhenjiang, 212013, China}
\author{Kai Kang}\sameaddress{1}
\author{Xiaoxin Wang}\sameaddress{1}
\author{Leyi Zhang}\sameaddress{1}
\author{Yanan Cai}\address{College of Science and Engineering, James Cook University, Queensland, 4814, Australia}\secondaddress{This work is done when Yanan Cai visits Western Sydney University}
\author{Changjing Zhuge $^*$}\sameaddress{1}
\author{Lei Zhang $^*$}\address{Department of Liver Surgery, Peking Union Medical College Hospital, Chinese Academy of Medical Sciences and Peking Union Medical College, Beijing 100730, China}

\date{}
\begin{abstract} 
Adaptive therapy improves cancer treatment by controlling the competition between sensitive and resistant cells through treatment holidays. This study highlights the critical role of treatment-holiday thresholds in adaptive therapy for tumors composed of drug-sensitive and resistant cells. Using a Lotka-Volterra model, adaptive therapy outcomes are compared with maximum tolerated dose therapy and intermittent therapy outcomes, showing that adaptive therapy success depends critically on the threshold for pausing and resuming treatment and on competitive interactions between cell populations. Three comparison scenarios between adaptive therapy and other therapies emerge: uniform-decline where adaptive therapy underperforms regardless of threshold, conditional-improve where efficacy requires threshold optimization, and uniform-improve where adaptive therapy consistently outperforms alternatives. Tumor composition including initial burden and resistant cell proportion influences outcomes. Threshold adjustments enable adaptive therapy to suppress resistant subclones while preserving sensitive cells, extending progression-free survival. Crucially, this work establishes an optimal control problem for time-to-progression and mathematically proves that under biological constraints like neutral competition or low initial burden, the theoretically optimal strategy is unrealizable as it requires infinitely many treatment holidays, rendering it clinically impractical. These findings emphasize personalized treatment strategies for enhancing long-term therapeutic outcomes.
\end{abstract}
%
%
\subjclass{92C50,92C42}
\keywords{adaptive therapy, treatment-holiday threshold, cancer dynamics, mathematical modelling, cell competition
}
\maketitle
\renewcommand{\thefootnote}{$\dagger$}
\footnotetext{These authors contributed equally.}
\renewcommand{\thefootnote}{*}
\footnotetext{Corresponding authors. zhuge@bjut.edu.cn (Changjing Zhuge); zhanglei44@pumch.cn (Lei Zhang)}
\renewcommand{\thefootnote}{\arabic{footnote}}
\section{Introduction}
Cancer remains one of the leading causes of death worldwide \cite{Bray2024}, with its biological complexity and heterogeneity posing significant challenges to conventional therapies such as chemotherapy, radiotherapy, and even immunotherapy \cite{Sung2021,Marusyk2012,Meacham2013a}. Therapeutic resistance, a central obstacle in cancer treatment, arises from the therapeutic selection pressure of pre-existing or induced resistant tumor cells \cite{Holohan2013,Sharma2010,Shi2023,Wang2019b,Tian2024}. This process drives the evolution of tumor tissue, ultimately leading to the dominance of primary or induced refractory subpopulations, and subsequent relapse and treatment failure \cite{Greaves2012,McGranahan2017,Turajlic2019}. The inherent heterogeneity and plasticity of tumor cells often render complete eradication exceedingly difficult, highlighting the necessity of strategies that address the entire tumor population and its evolutionary dynamics, rather than targeting specific subpopulations \cite{Meacham2013a,Gupta2011a,Marusyk2020}.

Recently, ecological theory has provided an innovative framework for understanding tumor progression, conceptualizing tumors as ecosystems where drug-sensitive and drug-resistant cells compete for resources \cite{Gatenby2009a,Gatenby2009,Brown2023,Aguade-Gorgorio2024a,Zahir2020,Bukkuri2023,Bilder2021,Hochberg2018,Brady2019}. Within this ecological framework, cancer cells under drug treatment activate regenerative or protective mechanisms to survive near-lethal stress, enabling them to enter a stem cell-like state and transmit resistance to their progeny \cite{Attolini2009,Vasan2019a}. This adaptation imposes costs, such as reduced proliferation rates, decreased environmental carrying capacity, and intensified interclonal competition \cite{Zhang2022e, Pressley2021}. 

The conventional maximum tolerated dose (MTD) therapy aggressively eliminates sensitive cell populations. However, this strategy inadvertently disrupts competitive equilibria, allowing resistant subclones to proliferate uncontrollably \cite{Maley2017,Enriquez-Navas2016,AthenaAktipis2015a,Gatenby2009,West2023}. This unintended consequence of MTD therapy suggests the limitations of conventional eradication-focused paradigms and highlights the need for therapeutic strategies that preserve sensitive cells to maintain ecological suppression of resistant populations \cite{Gatenby2009a,Gatenby2009,Gallaher2018,West2020a}.

Adaptive therapy (AT) is a dynamic treatment strategy that employs a feedback-controlled approach, adjusting therapeutic intensity and scheduling in response to real-time tumor burden \cite{Gatenby2009,Enriquez-Navas2015a,Enriquez-Navas2016}: treatment is paused when tumor burden declines to a certain lower threshold and resumed upon regrowth to an upper threshold. Adaptive therapy can be classified into dose-skipping (AT-S) and dose-modulating (AT-M) regimens, depending on whether more than two dose rates are used during treatment phases. By retaining a reservoir of sensitive cells, this intermittent dosing modality suppresses resistant subclone expansion and prolongs progression-free survival, as evidenced by clinical successes in prostate cancer and other malignancies.

Despite significant advancements \cite{West2023,Zhang2022e,Zhang2017d,McGehee2024.02.19.580916,Strobl2023b}, further research is needed to refine the optimization of treatment-holiday thresholds, a core parameter governing AT efficacy, as most existing studies focus on the efficiency of adaptive therapy under predefined thresholds. Several studies have explored the effects of varying thresholds on AT outcomes across different scenarios \cite{Liu2024d,Strobl2023b,Wang2024c,McGehee2024.02.19.580916,Tan2024,Strobl2024}, providing valuable insights that underscore the importance of threshold adjustments and establishing a solid foundation for further investigation. However, a deeper understanding of the dynamic mechanisms driving context-dependent outcomes under heterogeneous physiological conditions remains essential. It requires a systematic analysis of the underlying dynamical mechanisms by which different thresholds give rise to distinct effects, enhancing conventional empirical threshold selection to achieve an optimal balance between tumor suppression and resistance mitigation. Moreover, it can further facilitate the development of personalized strategies tailored to different conditions. Consequently, quantitative modelling of threshold-driven tumor evolutionary dynamics is essential for advancing precision in adaptive therapy design.

Therefore, this study leverages a mathematical model of tumor evolutionary dynamics to analyze and explain the effects of the thresholds on the outcomes of AT under various conditions of competitive interactions between sensitive and resistant subpopulations, elucidating the biophysical principles underlying threshold-mediated resistance control. The findings are expected to provide theoretical validation for personalized adaptive therapy protocols, facilitating a paradigm shift toward dynamically optimized cancer management. 

In summary, this study, as a state-of-the-art work, systematically investigates the impact of treatment-holiday thresholds on adaptive therapy outcomes using a Lotka–Volterra model that simulates the competitive dynamics between drug-sensitive and drug-resistant tumor cell populations. By quantifying the relationship between treatment-holiday threshold and employing time-to-progression (TTP) as an indicator of progression-free survival, the analysis demonstrates that AT efficacy critically depends on the competition between two subpopulations of cancer cells. Strong competition can allow for the indefinite delay of disease progression,
while weak competition requires the benefits of AT. Furthermore, this work also reveals that AT is not uniformly superior to MTD, as its efficiency depends on tumor composition (initial conditions) and patients' personal conditions (parameters). By comparing the TTP between AT and MTD, three distinct scenarios are identified: 1) \textbf{uniform-decline}, where AT consistently underperforms MTD regardless of the threshold; 2) \textbf{conditional-improve}, where the effectiveness of AT depends on the specific threshold selected; and 3) \textbf{uniform-improve}, where AT consistently outperforms MTD. These findings highlight the necessity of precise threshold tuning to optimize treatment outcomes. Collectively, these findings provide a theoretical foundation for refining empirical threshold selection and advancing personalized adaptive therapy protocols, thereby paving the way for dynamically optimized cancer management.

\section{MODEL DESCRIPTION}
To investigate the effects of different treatment strategies, the tumor dynamics under therapeutic conditions are modelled by the widely used Lotka-Volterra (LV)-type equations \cite{Strobl2021a,Aguade-Gorgorio2024a,McGehee2024.02.19.580916,Mathur2022,Beckman2020,Gatenby919,Gallagher2023.04.28.538766} because the cancer cells behave more similarly to unicellular organisms than to normal human cells \cite{Brown2023,Ermini2023,Pennisi2018,Deng2024}. The cancer cells are further assumed to be classified into two types: drug-sensitive cells ($S$), which can be killed by drugs, and drug-resistant cells ($R$), which cannot be killed by drugs. These two cell types compete within the cancer microenvironment (Figure \ref{fig:T:cartoon}A), and therefore the LV model is described as the following differential equations \eqref{eq:core-model:1}-\eqref{eq:core-model:3}. 
\begin{eqnarray}
    \dfrac{\mathrm{d}S}{\mathrm{d}t} &=& r_{S}\left(1-\dfrac{S+\alpha R}{K}\right)\left(1-\gamma\dfrac{D(t)}{D_0}\right)S-d_{S}S, \label{eq:core-model:1}\\
    \dfrac{\mathrm{d}R}{\mathrm{d}t} &=& r_R\left(1-\dfrac{\beta S+R}{K}\right)R-d_RR, \label{eq:core-model:2}\\
    N(t) &=& S(t)+R(t), \label{eq:core-model:3}
\end{eqnarray}
where $D(t)$ denotes the dose curve of therapeutic agents and $S$ and $R$ represent the burdens of the two types of cells respectively. Equation \eqref{eq:core-model:1} describes the growth and death of sensitive cells, which means that in the absence of therapy, their proliferation is regulated by intrinsic rate $ r_S $, logistic competition $\left( 1 - \dfrac{S + \alpha R}{K} \right)$ and the death rate is $d_S$. The drug-dependent suppression of sensitive cells is represented by the factor $\left( 1 - \gamma\dfrac{D(t)}{D_0} \right)$ where $\gamma$ is the killing strength of therapeutic agents on sensitive cells. Although this factor is multiplied with the proliferation terms, it is the integrated effective term taking both suppression of proliferation and increasing death rate for simplicity. Equation \eqref{eq:core-model:2} governs resistant cells, which obey logistical growth $\left( 1 - \dfrac{\beta S + R}{K} \right)$ at rate $ r_R $ and die at constant rate $ d_R $, unaffected by the drug. $K$ is the carrying capacity of the cancer microenvironment, indicating the maximum number of cells that can be supported. 

Moreover, the model equations \eqref{eq:core-model:1}–\eqref{eq:core-model:2} incorporate asymmetric competition between drug‐sensitive and drug‐resistant cells, characterized by the parameters $\alpha$ and $\beta$. Here, $\alpha$ quantifies the inhibitory effect of drug‐resistant cells on the growth of sensitive cells, while $\beta$ represents the impact of drug‐sensitive cells on the growth of resistant cells. Specifically, a higher value of $\alpha$ implies that resistant cells more strongly reduce the growth potential of sensitive cells, possibly due to resource competition or indirect suppression via secreted factors (e.g., growth factors, cytokines). Similarly, a higher value of $\beta$ indicates that sensitive cells exert a stronger influence on the growth of resistant cells, potentially through competition for space, nutrients, or other resources within the tumor microenvironment.

These competitive interactions determine the dominance relationships between the cell populations. For example, in the absence of therapy, coexistence of the two populations is possible under weak competition (i.e., $\alpha < 1$ and $\beta < 1$), whereas strong competition leads to bistability \cite{McGehee2024.02.19.580916,Aguade-Gorgorio2024a,Chapman2014,Aguade-Gorgorio2024b}. Specifically, if the resistant cell population is competitively dominant (i.e., $\alpha > 1$ and $\beta < 1$), the sensitive cell population will eventually become extinct; conversely, if $\alpha < 1$ and $\beta > 1$, the resistant cell population will be eliminated \cite{McGehee2024.02.19.580916}. In addition, the competition coefficients also significantly influence treatment outcomes, with a competitive advantage for sensitive cells being crucial for eradicating resistance \cite{Masud2022,Robertson-Tessi2024,Masud2022,Emond2023}. Thus, it is necessary to take these competitive dynamics into account in order to design better strategies.

\begin{figure}[phtb]
\centering
\includegraphics[width=\textwidth]{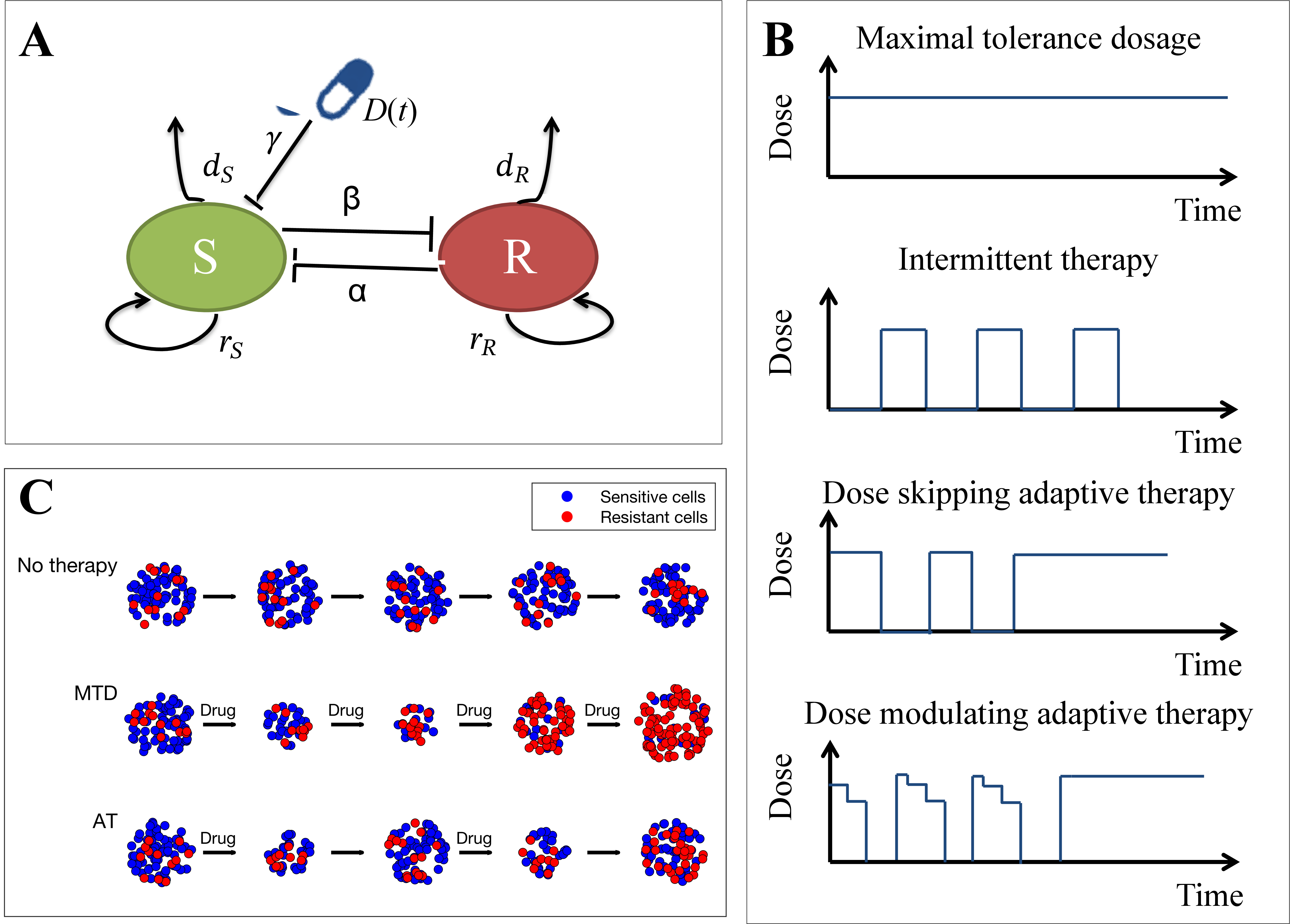}
\caption{Schematic diagram of cancer dynamics models and illustrations of different treatment strategies and their outcomes.
\textbf{(A)} A schematic representation of the Lotka-Volterra competition between sensitive cancer cells ($S$) and resistant cancer cells ($R$). The interactions between these two cell types are influenced by the carrying capacity $K$ and the competition coefficients $\alpha$ and $\beta$. Additionally, in the presence of therapeutic agents, sensitive cells are killed by the drugs.
\textbf{(B)} Schematic illustration of four typical therapeutic strategies: the maximal tolerance dosage (MTD), intermittent therapy (IT), dose-skipping adaptive therapy (AT-S), and dose-modulating adaptive therapy (AT-M).
\textbf{(C)} Schematic illustration of tumor burden and compositional evolution under distinct therapeutic strategies.}
\label{fig:T:cartoon}
\end{figure}

Different therapeutic strategies are characterized by distinct profiles of the dosing function $D(t)$. Two widely used therapeutic approaches and two regimens of adaptive therapy are considered in this work, including maximum tolerated dose, intermittent therapy (IT), dose-skipping adaptive therapy (AT-S) and dose-modulating adaptive therapy (AT-M) (Figure \ref{fig:T:cartoon}B). The conventional maximum tolerated dose approach seeks to maximize cancer cell eradication by administering the highest dose that a patient can tolerate based on acceptable toxicity levels, which can be defined by continuous administration as \eqref{eq:D:MTD}.
\begin{flalign}
\text{(MTD)} && D(t) = D_0, \quad \mbox{for all }\quad t.&&\label{eq:D:MTD}
\end{flalign}
As continuous high-intensity dosage leads to the development of therapeutic resistance through Darwinian selection pressures \cite{Enriquez-Navas2016}, as well as that high dosage of therapeutic agents causes severe toxicity, the intermittent therapy regimens were proposed, which strategically suspend drug administration during predefined recovery periods and emerged as a clinically utilized strategy for balancing therapeutic efficacy with toxicity management \cite{Mathur2022}. The IT strategies can be characterized as deterministic treatment cycles as \eqref{eq:D:IT}.
\begin{flalign}
    \text{(IT)}&&
    D(t) = \begin{cases}
                D_0, & nT\leq t \leq nT+T_D, \\
                0, & \mbox{otherwise}.
            \end{cases}
    &&
    \label{eq:D:IT}
\end{flalign}
where $T$ is the fixed period of one cycle of treatment, $T_D$ is the duration of drug administration and $n=1,2,\cdots$ represents the number of cycles.

As described before, AT dynamically adjusts cancer treatment based on tumor response to prolong efficacy and manage resistance. Two primary approaches within adaptive therapy are dose-skipping and dose-modulating regimens. In dose-skipping adaptive therapy, treatment is administered at the maximum tolerated dose until the tumor burden shrinks to a predetermined threshold ($C_\mathrm{TH0}$), for example, $50$\% of the baseline burden; therapy is then paused, allowing the tumor to regrow to a specific threshold ($C_\mathrm{TH1}$), usually, $100$\% of the baseline, before resuming treatment. In dose-modulating adaptive therapy, treatment doses are adjusted at regular intervals based on tumor response: increasing the dose if the tumor grows and decreasing it if the tumor shrinks. This strategy seeks to tailor therapy dynamically to tumor behavior, potentially reducing toxicity and managing resistance more effectively. Both approaches exploit the competitive dynamics between treatment-sensitive and resistant cancer cells to enhance long-term treatment outcomes and have already succeeded in trials or experiments \cite{Zhang2017d,Enriquez-Navas2016,Smalley2019a,Seyedi2024,West2023}. Based on the above description, the two AT regimens can be formulated as following equations \eqref{eq:D:AT-S} and \eqref{eq:D:AT-M} respectively.
\begin{flalign}
    \text{(AT-S)}&&
    D(t)=\begin{cases}
        D_0,  &  \mbox{if } N(t) \geq C_\mathrm{TH0} N_0 \text{ until } N(t) \geq C_\mathrm{TH1}N_0,  \\
        0,  & \mbox{if } N(t)\leq C_\mathrm{TH1} N_0\text{ and } N^{\prime}(t)>0 \mbox{ until } N(t)\geq C_\mathrm{TH1}N_0.
    \end{cases}\label{eq:D:AT-S}
    &&\\[0.1cm]
    \text{(AT-M)}&&
    D(t) = \begin{cases}
        (1+\delta_1)D_0, & \mbox{if } N(t) \geq C_\mathrm{TH2}N_0, \\
        D_0, & \mbox{if } C_\mathrm{TH1} N_0 \leq N(t) \leq C_\mathrm{TH2}N_0 \mbox{ until } N(t) \leq C_\mathrm{TH1} N_0, \\
        (1-\delta_2)D_0, & \mbox{if }C_\mathrm{TH0}N_0 \leq N(t) \leq C_\mathrm{TH1}N_0 \mbox{ until } N(t) \leq C_\mathrm{TH0}N_0. \\
        0, & \mbox{if } N(t) \leq C_\mathrm{TH2}N_0 \mbox{ and } N'(t)>0 \mbox{ until } N(t)\geq C_\mathrm{TH2}N_0.
    \end{cases}
    \label{eq:D:AT-M}
    &&
\end{flalign}
where $N_0=S_0+R_0$ is the initial tumor burden and $N'(t)$ is the derivative of $N(t)$ representing whether the tumor burden is growing or shrinking. The dose-modulating thresholds satisfy $C_\mathrm{TH0}<C_\mathrm{TH1}<C_\mathrm{TH2}$ and $C_\mathrm{TH0}<1<C_\mathrm{TH2}$. Notably, the dosing strategies \eqref{eq:D:AT-S} and \eqref{eq:D:AT-M} represent state-dependent switching mechanisms rather than closed-form mathematical definitions. This implementation induces history-dependent therapeutic decision and consequently, the governing equations \eqref{eq:core-model:1}-\eqref{eq:core-model:3} constitute delay differential equations with history-dependent terms $D(t)$.

As a state-of-the-art investigation, this study provides comprehensive parameter analysis rather than case-specific parameter estimation. So, in this work, the parameter values are taken in ranges according to established experimental and clinical studies \cite{Zhang2017d,McGehee2024.02.19.580916,West2019b,Sunassee2019,Grassberger2019a}. The proliferation rates of sensitive and resistant cells, i.e. $r_S$ and $r_R$, are derived from \textit{in vitro} experimental data \cite{Zhang2017d}, while the drug-induced killing strength $\gamma$ is sourced from the estimation based on several experiments of cell lines\cite{West2019b}. Given the critical clinical risk associated with large tumor burdens, the initial tumor burden in this study is set as 75\% of the environmental carrying capacity \cite{Sunassee2019}. The initial proportion of resistant cells is taken from existing dynamic models \cite{Grassberger2019a}. All the baseline parameter values are summarized in Table \ref{tab:parameters}.

\begin{table}[tbph]
\caption{\label{tab:parameters}Model parameters and their ranges.}
\begin{tabular}{l p{18em} l l}
\hline
Parameter & Description & Value or range & Reference \\ \hline
$K$ & Total environmental carrying capacity & $10000$ & \\\hline
$\alpha$ & Competition coefficient of resistant cells on sensitive cells & $(0,+\infty)$ & \\\hline
$\beta$ & Competition coefficient of sensitive cells on resistant cells & $(0,+\infty)$ & \\\hline
$r_S$ & Proliferation rate of sensitive cells & $0.035$ & \cite{Zhang2017d} \\\hline
$r_R$ & Proliferation rate of resistant cells & $0.027$ & \cite{Zhang2017d} \\\hline
$d_S$ & Death rate of sensitive cells & $0.0007$ & \cite{Zhang2017d} \\\hline
$d_R$ & Death rate of resistant cells & $0.00054$ & \cite{Zhang2017d} \\\hline
$\gamma$ & The killing strength of therapeutic agents & $1.5$ & \cite{West2019b} \\\hline
$D_0$ & Maximal tolerance dose & $1$ \\
[0.3cm]\hline
$n_0=\dfrac{N_0}{K}$ & Initial tumor burden& $[0,0.75]$ & \cite{Sunassee2019} \\
[0.3cm]\hline
$f_{R}=\dfrac{R_0}{R_0+S_0}$ & Initial proportion of resistant cells & $[0.001,0.01]$ & \cite{Grassberger2019a} \\\hline
$C_\mathrm{TH0}$, $C_\mathrm{TH1}$&\multicolumn{3}{l}{The treatment-holiday thresholds for dose-skipping adaptive therapy}\\\hline
$C_\mathrm{TH0}$, $C_\mathrm{TH1}$, $C_\mathrm{TH2}$ & \multicolumn{3}{l}{The dose-modulating thresholds for dose-modulating adaptive therapy}\\\hline
\end{tabular}
\end{table}
\section{RESULTS}
\subsection{The treatment-holiday threshold is crucial for the efficiency of adaptive therapy}
To investigate the effect of treatment-holiday thresholds on the outcomes of adaptive therapy, a simplified assumption of neutral competition ($\alpha=\beta=1$) is considered. The outcomes are primarily characterized by the time to progression (TTP), which is typically defined as the time at which the tumor burden grows to a certain percentage greater than its initial burden\cite{Zhang2017d,McGehee2024.02.19.580916,West2019b,Sunassee2019,Grassberger2019a}. For simplicity, this percentage is set to $120$\% in this study.

As illustrated in the example of tumor evolution under adaptive therapy (Figure~\ref{fig:T:equal-1}A), during the initial phase of adaptive therapy, drugs are administered to suppress tumor growth primarily through the depletion of sensitive cells. During this stage, the proportion of sensitive cells is high, which aligns with real-world cases. Following the principles of adaptive therapy, treatment is paused when the tumor burden reaches a predefined treatment-holiday threshold $C_{\mathrm{TH0}}N_0$, conventionally set at $50$\%  of the initial tumor burden, i.e., $C_{\mathrm{TH0}}=0.5$. The treatment holidays allow sensitive cells to repopulate, leading to tumor regrowth. When the tumor burden returns to the initial level $N_0$, treatment resumes. The treatment-holiday threshold $C_\mathrm{TH0}$ thus defines a cycle, from the start of treatment to its suspension. These cycles continue until resistant cells dominate the tumor burden, ultimately leading to disease progression.

To further elucidate this phenomenon, phase-plane analysis reveals distinct dynamical regimes with and without treatment (Figure~\ref{fig:T:equal-1}B1-B2). In untreated systems, tumor populations evolve toward coexistence equilibria, as indicated by the nullcline, which mathematically corresponds to the diagonal in Figure~\ref{fig:T:equal-1}B1. In the presence of treatment (Figure~\ref{fig:T:equal-1}B2), the long-term steady state shifts to the upper-left corner, where resistant cells dominate the tumor, and sensitive cells eventually disappear. From a dynamical perspective, treatment exerts selective pressure, pushing the tumor toward a resistant cell-dominant state. Therefore, the significance of treatment holidays lies in allowing the tumor to remain longer in a state dominated by sensitive cells, thereby delaying the transition to resistance. In this regard, increasing the treatment-holiday threshold can slow the transition of tumor composition to a resistant state, which implies it is possible to improve adaptive therapy outcomes by modulating the threshold and preventing continuous treatment from driving the tumor toward resistance via Darwinian selection \cite{Gatenby2009}.

\begin{figure}[phtb]
\centering
\includegraphics[width=\textwidth]{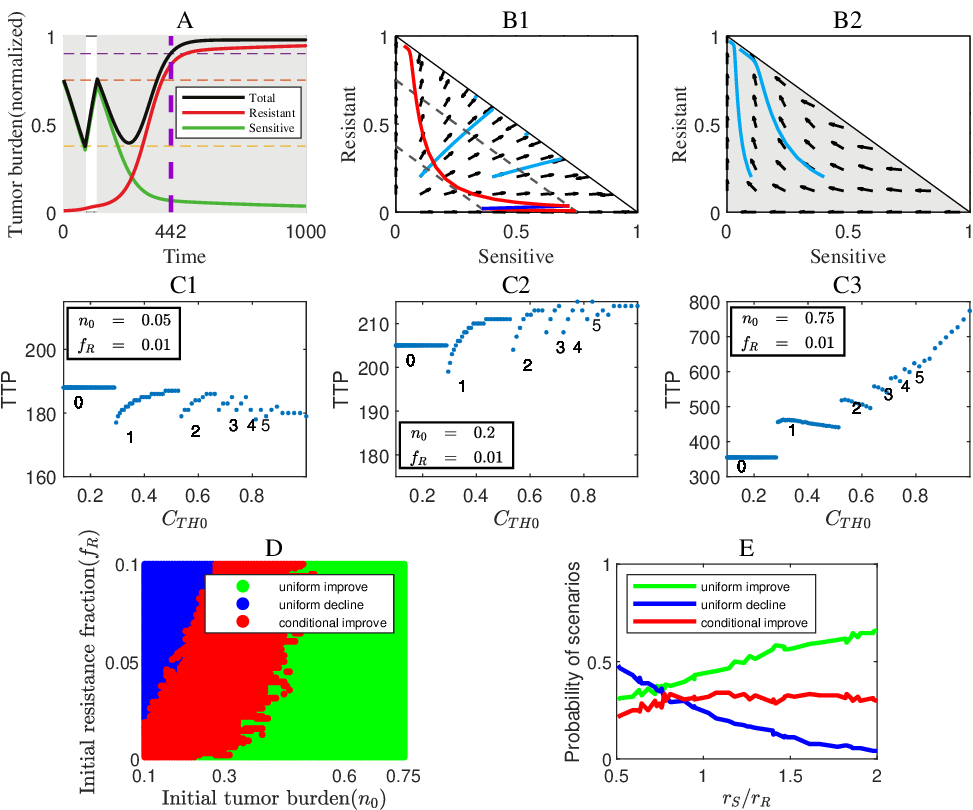}
\caption{Threshold-mediated therapeutic outcomes under neutral competition ($\alpha = \beta = 1$).
\textbf{(A)} A representative example of tumor evolution under adaptive therapy (AT), where the initial proportion of resistant cells $f_R = 0.01$ and the initial tumor burden as a fraction of carrying capacity is $n_0 = 0.75$. The purple dashed line represents the time to progression (TTP) under the AT strategy with the threshold $C_\mathrm{TH0} = 0.5$; the red, green, and black lines represent resistant cells, sensitive cells, and total tumor burden, respectively. Shaded areas indicate periods of active treatment. 
\textbf{(B1-B2)} The phase diagram and the vector field of tumor evolution with and without treatment. Light blue trajectories illustrate the trajectories with different tumor initial states. The upper right blank region represents situations beyond environmental carrying capacity ($K$). Panel (B1) illustrates the dynamics without treatment, while panel (B2) shows the cases with maximum tolerated dose strategy. The dark blue trajectory in (B1) shows the dynamics of the total tumor burden without treatment, while the red curve illustrates dynamics under AT as shown in panel (A). 
\textbf{(C1)-(C3)} The impact of changing the treatment-holiday threshold on TTPs under different initial tumor compositions. The numbers on the plots indicate the total number of treatment cycles experienced before disease progression, with $0$ cycles of AT corresponding to the MTD strategy. Panel (C1) represents a \textbf{uniform-decline} scenario of AT, where the TTPs with AT-S are consistently worse than those with MTD; panel (C2) shows a \textbf{conditional-improve} scenario, where AT-S outperforms MTD only for certain threshold values; and panel (C3) demonstrates a \textbf{uniform-improve} scenario, where all AT outperforms MTD for all threshold values.
\textbf{(D)} The dependence of the three scenarios on the initial tumor burden and composition.
\textbf{(E)} Probability of the three outcome scenarios versus the proliferative fitness ratio $r_S/r_R$ with randomly initial tumor states, demonstrating that a sensitive cell growth advantage ($r_S/r_R > 1$) enhances the AT success probability.
}
\label{fig:T:equal-1}
\end{figure}

Therefore, the effects of varying the treatment-holiday threshold are examined under different initial tumor states, characterized by the tumor burden ($n_0$) and the fraction of resistant cells ($f_R$) (ref. Table~\ref{tab:parameters}). As shown in Figures~\ref{fig:T:equal-1}C1-C3, altering the treatment-holiday threshold significantly impacted TTP. Notably, when the threshold is too low, meaning the tumor never reached the treatment-holiday threshold even during the initial cycle, adaptive therapy lacks treatment holidays and effectively becomes equivalent to MTD therapy. Consequently, as $C_{\mathrm{TH0}}$ is around $0$, TTP of AT is de facto that of MTD therapy, as shown in the leftmost segments of Figures~\ref{fig:T:equal-1}C1-C3.

Importantly, for different initial tumor states, three distinct scenarios emerged when comparing adaptive therapy with MTD. The first scenario, \textit{uniform-decline} (Figure~\ref{fig:T:equal-1}C1), occurs when adaptive therapy consistently performs worse than MTD, regardless of the threshold. This scenario typically arises in tumors with a low initial tumor burden. The second scenario, \textit{conditional-improve} (Figure~\ref{fig:T:equal-1}C2), occurs when adaptive therapy outperforms MTD only at specific threshold values, suggesting that careful threshold selection is required. The third scenario, \textit{uniform-improve} (Figure~\ref{fig:T:equal-1}C3), occurs when adaptive therapy is always superior to MTD.

The findings highlight the complex relationship between treatment-holiday thresholds and TTP. Moreover, as shown in Figures~\ref{fig:T:equal-1}C2-C3, the effect of the threshold is neither monotonic nor continuous, exhibiting discontinuities as the number of treatment cycles changes discretely (Figure \ref{fig:S:supplement_cellcount}). Specifically, in tumors with a low initial burden (Figures~\ref{fig:T:equal-1}C1 and C2), increasing the number of treatment holidays reduces drug exposure, which in turn weakens tumor suppression because under these conditions, the initial tumor burden is low, and therefore, environmental resources remain abundant which makes the changes in the burden of sensitive cells have little impact on the growth rate of resistant cells. In contrast, for tumors with a higher initial burden (Figure~\ref{fig:T:equal-1}C3), where environmental resources are constrained, fluctuations in sensitive cell numbers significantly affect the growth of resistant cells. In such cases, increasing the number of treatment holidays enables sensitive cells to grow and exert stronger suppression on resistant cells, leading to a marked increase in TTP. This aligns with clinical observations in prostate cancer, where intermittent dosing preserved therapeutic sensitivity \cite{Zhang2017d}.

To further investigate the impact of initial tumor states on these three scenarios, we examined how different initial tumor burdens and initial resistant cell fractions influence adaptive therapy outcomes (Figure~\ref{fig:T:equal-1}D). When the initial tumor burden exceeds $0.6K$, adaptive therapy uniformly improves TTP, regardless of the initial proportion of resistant cells. However, when the initial tumor burden is below $0.3K$, adaptive therapy can only conditionally improve TTP, implying that careful selection of the treatment-holiday threshold is crucial. In the worst scenario, when the initial tumor burden is very low and the proportion of resistant cells is relatively high (Figure~\ref{fig:T:equal-1}D, top-left), adaptive therapy fails to improve TTP at any threshold, indicating primary resistance. These findings extend prior theoretical work \cite{Strobl2021a} by quantifying phase-transition-like phenomena in threshold-mediated outcomes.

Furthermore, clinical implementation faces inherent uncertainty in tumor state assessment and threshold selection, implying that both the treatment-holiday threshold and the initial tumor burden and resistance fraction exhibit random variability. To address this, Monte Carlo simulations were performed to evaluate the probabilities of the three outcome scenarios under parameter variability (Figure~\ref{fig:T:equal-1}E). Inspired by the Norton-Simon Hypothesis \cite{Traina2011,Simon2006}, which posits that sensitive cells typically have a greater fitness advantage than resistant cells (i.e., sensitive cells have a higher intrinsic proliferation rate than resistant cells), we further explored the effect of the relative growth rate of sensitive cells ($r_S/r_R$) on the probabilities of the three scenarios. By randomly varying the treatment threshold and, while keeping $r_S/r_R$ fixed, randomly altering $r_S$ and $r_R$, we found that the probability of the \textit{uniform-improve} scenario increases as the relative fitness of sensitive cells increases. When $r_S/r_R$ is small, adaptive therapy is less likely to improve TTP relative to MTD. Conversely, when $r_S/r_R$ is high, indicating that sensitive cells have a significant growth advantage over resistant cells, adaptive therapy is more likely to enhance TTP. These results confirm that the treatment-holiday threshold has a profound impact on adaptive therapy outcomes and patient-specific tumor characteristics, including initial tumor burden and the relative fitness of sensitive and resistant cells, significantly influence treatment outcomes. These findings align with existing research \cite{Zhang2022e} and demonstrate that $r_S/r_R$ dependence provides quantitative criteria for patient stratification, providing insights for addressing a key challenge in empirical threshold selection \cite{West2023}.

\subsection{Context-dependent superiority of adaptive therapy across treatment strategies}
This subsection further compares the efficacy of adaptive therapy and intermittent therapy under the uniform-improve scenario, extending the previous comparison with maximum tolerated dose. As indicated in prior results, the effectiveness of different treatment strategies is influenced by treatment parameters. Therefore, to ensure a fair comparison, treatment parameters must be appropriately aligned. A reasonable approach is to compare strategies based on their best achievable outcomes (Figure~\ref{fig:T:equal-3}). Specifically, for intermittent therapy, the optimal treatment cycle and duration should be selected, while for adaptive therapy, the optimal threshold values should be determined.

The optimal parameter combination is determined based on different initial tumor states. Since this study does not focus on identifying these optimal parameters of intermittent therapy, computational details are omitted. However, as previous results suggest that under adaptive therapy in the conditional-improve and uniform-improve scenarios, a threshold closer to $1$ generally leads to a longer TTP (Figure~\ref{fig:T:equal-3}). In practical applications, a threshold near 1 would necessitate frequent initiation and cessation of treatment, requiring highly precise monitoring of tumor status, which is unrealistic. Similarly, for intermittent therapy, frequent treatment pauses and resumptions are impractical. Given these practical constraints, achieving the true optimal outcome is infeasible. Consequently, only suboptimal treatment outcomes that closely approximate the optimal results are considered as benchmarks for comparison.

\begin{figure}[phtb]
\centering
\includegraphics[width=\textwidth]{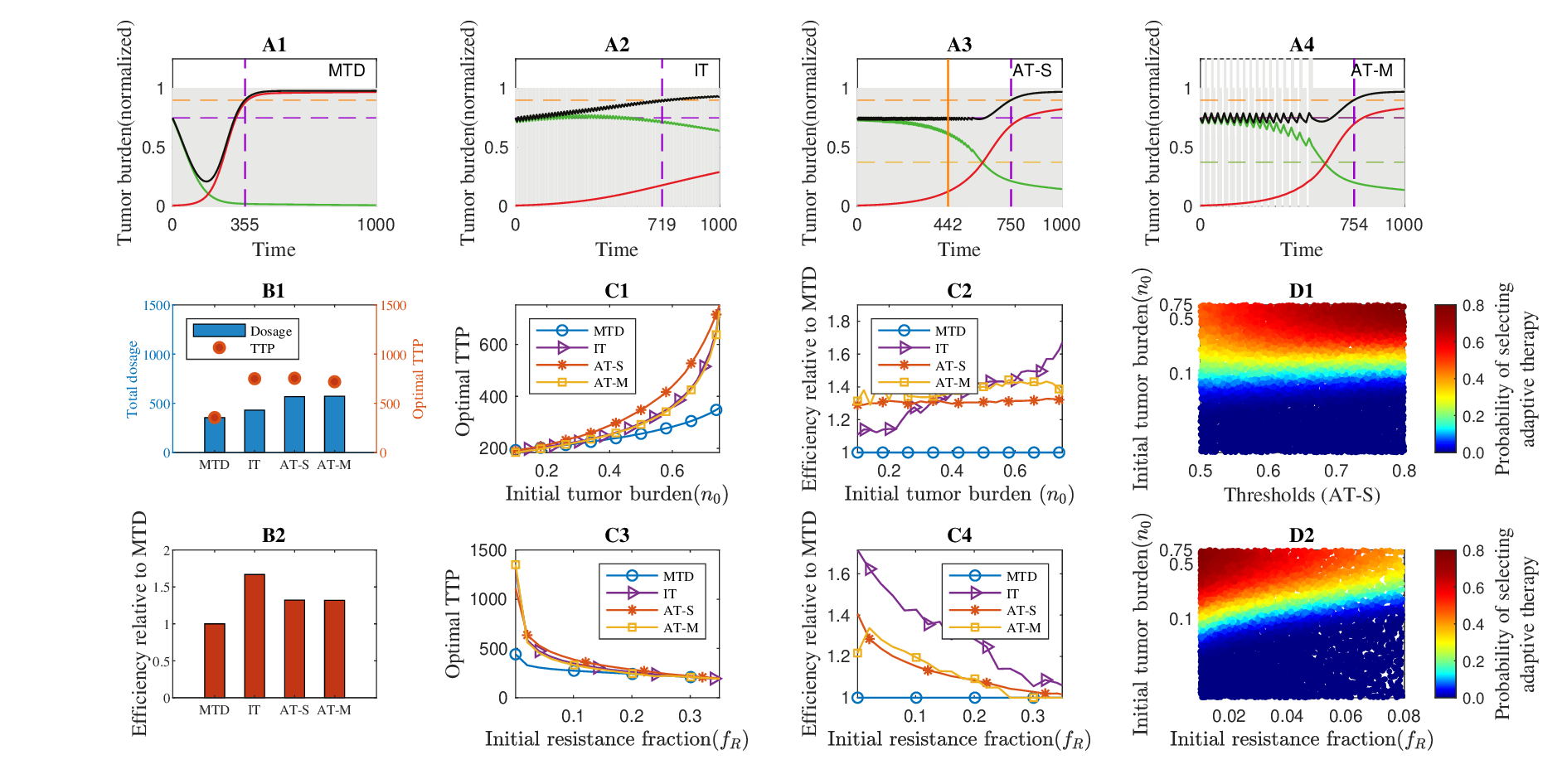}
\caption{Comparison of optimal outcomes under different therapeutic strategies and success landscapes of adaptive therapy. 
\textbf{(A1)}-\textbf{(A4)} Tumor evolution dynamics under four therapeutic strategies with optimal treatment-holiday thresholds, given $n_0 = 0.75$ and $f_R=0.01$. The purple dashed line indicates the TTP with optimal treatment parameters, while the orange solid line in (A2) represents the TTP with the conventional 50\% threshold ($C_\mathrm{TH0}=0.5$). Red, green, and black curves represent resistant cells, sensitive cells, and total tumor burden, respectively. Shaded regions denote periods of active treatment. For IT, the treatment interval ($T$) is $15$, and the treatment duration ($T_D$) is $6$. For AT-S, the optimal threshold is close to $1$, so here $C_\mathrm{TH0}$ is set to $0.98$. Similarly, for AT-M, the middle threshold is $C_\mathrm{TH1} = 1$, and the upper threshold is $C_\mathrm{TH2} = 1.05$.
\textbf{(B1)} Cumulative drug toxicity and the corresponding TTP for the four therapeutic strategies. The cumulative toxicity is quantified by the total drug dosage administered before disease progression, given by equation \eqref{eq:total dosage and efficiency}.
\textbf{(B2)} Therapeutic efficiency index, defined as the TTP-to-toxicity ratio as equation \eqref{eq:total dosage and efficiency}. Here the efficiency indices are displayed as normalized to that of MTD.
\textbf{(C1)}-\textbf{(C2)} Effects of varying the initial tumor burden on TTP and the efficiency of the four treatment strategies.
\textbf{(C3)}-\textbf{(C4)} Effects of varying the initial proportion of resistant cells on TTP and the efficiency of the four treatment strategies.
\textbf{(D1)}-\textbf{(D2)} Success probability landscapes of AT-S as a function of the initial tumor state and treatment threshold. The color at each point represents the probability of AT-S improving treatment outcomes compared to MTD. (D1) illustrates the dependence of success probability on the initial tumor burden and treatment-holiday threshold with randomly varying $f_R$, $r_S$, and $r_R$, which simulates the uncertainty in tumor proliferation and treatment decisions; while (D2) shows its dependence on the initial fraction of resistant cells and initial tumor burden with randomly varying $C_\mathrm{TH0}$, $r_S$, and $r_R$.}
\label{fig:T:equal-3}
\end{figure}

For intermittent therapy, the treatment cycle period ($T$) is examined to range from $7$ to $63$ time units, with drug administration lasting from $7$ to $21$ time units ($T_D$) per cycle. Since this study does not focus on identifying these optimal parameters of IT, computational details are omitted. The result shows that the suboptimal period is $T=15$ with $T_D=9$, simulating a regimen in which the patient undergoes both treatment and rest periods within a feasible schedule (Table~\ref{tab:parameters} and equations~\eqref{eq:D:IT}). For dose-skipping adaptive therapy, the treatment-holiday threshold is set to $C_\mathrm{TH0}=0.98$, while for dose-modulating adaptive therapy, $C_\mathrm{TH1} = 1$ and $C_\mathrm{TH2} = 1.05$ are selected as the suboptimal treatment parameters (Figure~\ref{fig:T:equal-3}A1-A4).

In the uniform-improve scenario for adaptive therapy, a comparison of the suboptimal outcomes across the four treatment strategies (Figure~\ref{fig:T:equal-3}A1-A4) reveals that both adaptive therapy strategies result in a longer TTP than intermittent therapy, while intermittent therapy still outperforms MTD. Furthermore, when comparing the two adaptive therapy strategies, AT-M slightly outperforms AT-S. However, considering the operational complexity of AT-M, AT-S is primarily considered in subsequent discussions to balance efficacy with practical feasibility.

From an efficiency perspective, due to factors such as drug side effects, a higher total drug dosage is not necessarily preferable. Therefore, based on the principle of maximizing TTP with the minimum amount of drug, the contribution of unit dosage to TTP is considered a measure of treatment strategy efficiency (equation~\eqref{eq:total dosage and efficiency}, Figure~\ref{fig:T:equal-3}B1-B2):

\begin{equation}
    \text{Total dosage} = \int_0^{\mathrm{TTP}} D(t) \, \mathrm{d}t, \quad \text{Efficiency} = \frac{\text{TTP}}{\text{Total dosage}} \label{eq:total dosage and efficiency}
\end{equation}

It is observed that MTD requires the least total drug dosage (Figure~\ref{fig:T:equal-3}B1), but this is primarily due to its short TTP. Consequently, MTD's efficiency is the lowest among the four treatment strategies (Figure~\ref{fig:T:equal-3}B2). Notably, when efficiency is used as the evaluation criterion, intermittent therapy performs the best (Figure~\ref{fig:T:equal-3}B2). This is because intermittent therapy, compared to adaptive therapy, has longer treatment holidays with a relatively smaller decrease in TTP, resulting in a higher contribution of unit drug dosage to TTP (Figure~\ref{fig:T:equal-3}B2). This suggests that intermittent therapy is a reasonable option, particularly for patients with lower drug tolerance, which is consistent with clinical practice.

Furthermore, to compare the robustness of the four treatments, the stability of TTP and treatment efficiency under uncertainties in patient physiological conditions and tumor status monitoring is examined. The impact of different initial tumor states on TTP and efficiency under the suboptimal parameters for the four strategies is analyzed (Figure~\ref{fig:T:equal-3}C1-C4). For different initial tumor burdens, the trends in TTP for the four treatment strategies under the suboptimal parameters are similar; TTP increases with the tumor burden, which is related to the criteria for determining tumor progression. Since progression is determined based on the initial tumor burden, a larger initial tumor burden necessitates a higher burden for progression to occur, leading to this seemingly counterintuitive result (Figure~\ref{fig:T:equal-3}C1). From this perspective, some researchers have proposed a ``modifying adaptive therapy'' strategy, which delays treatment until the tumor reaches a certain size \cite{Hansen2020}.

Moreover, intermittent therapy shows significant variability with different initial tumor burdens (Figure~\ref{fig:T:equal-3}C2). When the tumor burden is small, intermittent therapy is less efficient than adaptive therapy. As the tumor burden increases, this relationship reverses. Therefore, the robustness of intermittent therapy’s efficiency is inferior to that of adaptive therapy. This suggests that when adopting intermittent therapy, a more precise assessment of the patient's specific condition is necessary.

On the other hand, both adaptive therapy strategies and intermittent therapy show similar trends in TTP and efficiency with varying initial proportions of resistant cells (Figure~\ref{fig:T:equal-3}C3-C4), indicating that the optimal TTP is more influenced by cell composition when the proportion of resistant cells is very low. Notably, when the initial proportion of resistant cells is high, the TTP and efficiency of different treatment strategies are close to those of MTD. This is because, when the initial proportion of resistant cells is high, the determining factor for TTP is no longer the competition between sensitive and resistant cells but rather the proliferation of resistant cells themselves, which reduces the differences between treatment strategies.

Furthermore, similar to the previous results, considering the uncertainty in real-world scenarios, the probability of adaptive therapy being the best strategy is explored under conditions of random tumor growth rates and status monitoring. The probability that adaptive therapy results in a longer TTP than both MTD and intermittent therapy is shown in Figure~\ref{fig:T:equal-3}D1-D2. In real scenarios, since the tumor growth rate is unknown and tumor composition is difficult to obtain, the dependence of TTP on treatment-holiday thresholds, initial tumor burden, and initial resistant cell proportion is explored under random variations of other factors. Figure~\ref{fig:T:equal-3}D1 shows the dependence of the probability of adaptive therapy outperforming other treatment strategies on different thresholds ($C_\mathrm{CH0}$ and initial tumor burden ($n_0$), with the other parameters, including the initial resistant cell proportion ($f_R$) and the intrinsic growth rates of tumor cells ($r_S$ and $r_R$) varying randomly. Figure~\ref{fig:T:equal-3}D2 shows the dependence of the probability on the initial resistant fraction ($f_R$) and the initial tumor burden ($n_0$), with other parameters, $C_\mathrm{TH0}$, $r_S$ and $r_R$, varying randomly. 

It can be observed that the probability of selecting adaptive therapy increases significantly when the initial tumor burden is larger. However, when the initial tumor burden is smaller, adaptive therapy does not consistently improve TTP, which aligns with previous findings. Notably, the selection of the treatment-holiday threshold and tumor composition has a weaker impact on the selection of adaptive therapy than the initial tumor burden, indicating that when detailed information about tumor components is unavailable, parameters such as the treatment threshold have less influence, with the initial tumor burden playing a key role. On the other hand, TTP is highly sensitive to the selection of the threshold, among other factors (Figures \ref{fig:T:equal-1} and \ref{fig:S:equal:supplement-1}-\ref{fig:S:equal:supplement-5}). In other words, the selection of the optimal threshold is also dependent on various tumor states. Therefore, when other parameters are randomly varied, the impact of the threshold on TTP becomes overshadowed by fluctuations in these parameters.

These findings are consistent with clinical observations, emphasizing the importance of a precise understanding of tumor dynamics, particularly intrinsic growth rates. If such tumor information is available, treatment strategy selection can be significantly improved, highlighting the critical importance of accurate tumor characterization, especially regarding tumor dynamics.

\subsection{Conditional superiority landscapes of adaptive therapy compared to other strategies under non-neutral competition
}
The competitive interaction between sensitive ($S$) and resistant ($R$) cells, governed by coefficients $\alpha$ and $\beta$, fundamentally determines tumor evolution dynamics under therapeutic stress \cite{Park2023c,Gatenby2018,Wang2024a,Parker2021}. Four distinct regimes emerge based on relative competition strengths: 1) weak competition ($\alpha < 1$, $\beta < 1$), 2) strong competition ($\alpha > 1$, $\beta > 1$), 3) asymmetric competition with $\alpha < 1 < \beta$, or 4) $\beta < 1 < \alpha$. Detailed analyses of $\alpha < 1 < \beta$ and $\beta < 1 < \alpha$ cases demonstrate that therapeutic strategies can induce resistant population extinction under specific conditions \cite{McGehee2024.02.19.580916}.

Under strong competition ($\alpha > 1$, $\beta > 1$), untreated tumors evolve toward exclusive dominance of either $S$ or $R$, depending on initial composition (Figure~\ref{fig:T:geq-1}A), mathematically equivalent to equations \eqref{eq:core-model:1}--\eqref{eq:core-model:3} possessing two stable steady states. Continuous MTD administration drives the tumor toward a stable steady state where sensitive and resistant cells maintain a fixed ratio with unchanged tumor burden (Figure~\ref{fig:T:geq-1}B). In the absence of treatment, the tumor's final state depends on its initial condition, mathematically determined by the basin of attraction in which the system starts. Thus, sensitive cell dominance at initiation leads to resistant cell elimination without treatment, requiring only therapeutic containment to prevent long-term progression. Conversely, initial resistant cell dominance under MTD results in a stable tumor burden with constant cellular proportions. The initial tumor burden determines whether disease progression occurs. Even when exceeding progression thresholds, tumor burden stabilizes below carrying capacity. This property renders the "uniformly-improve" scenario inapplicable for strong competition systems, as therapeutic outcomes primarily depend on initial tumor state rather than threshold calibration (Figure~\ref{fig:T:geq-1}C). That is, as long as the long-term steady state under treatment remains acceptable (Figure~\ref{fig:T:geq-1}B), treatment strategy selection becomes less critical for strong competition cases.

\begin{figure}[phtb]
\centering
\includegraphics[width=\textwidth]{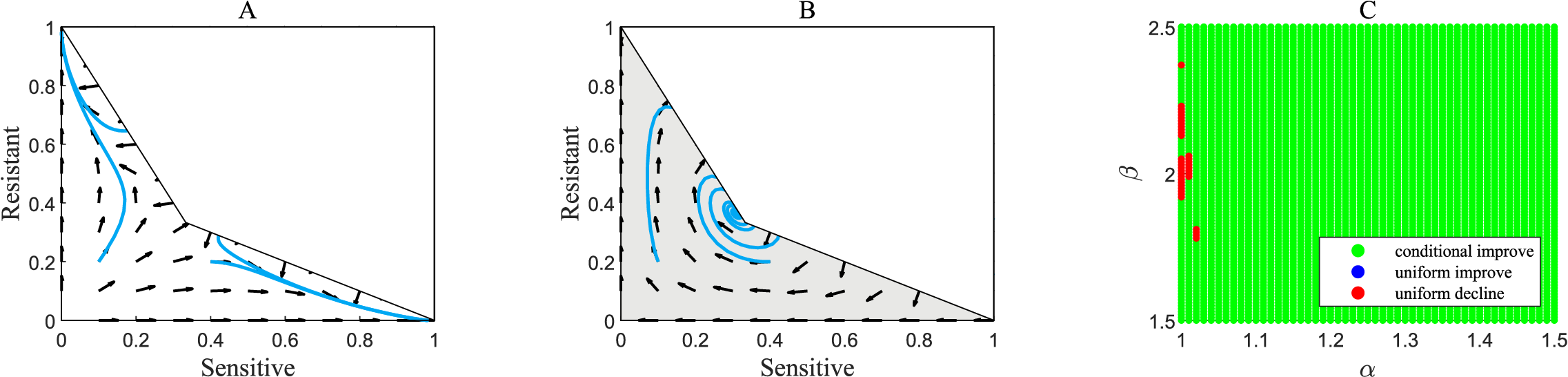}
\caption{Tumor evolution characteristics under strong competition in the presence of therapy.
\textbf{(A)} Phase diagram and vector fields of tumor evolution without treatment. The blue lines indicate that the long-term tumor state depends on the initial conditions. The blank region indicates that the effective tumor burden exceeds the carrying capacity ($S+\alpha R>K$ or $\beta S+R>K$).
\textbf{(B)} Phase diagram and vector fields of tumor evolution with treatment. The long-term tumor dynamics converge to a state characterized by a heterogeneous composition of sensitive and resistant cells. The other features are similar to those in panel (A).
\textbf{(C)} The dependence of the three scenarios on the competition coefficients $\alpha$ and $\beta$, where, in fact, only two scenarios occur: uniform-decline (red) and conditional-improve (green), with the uniform-improve scenario absent. Here the initial tumor burden is taken as $n_0=0.75$ and the fraction of resistant cells is taken as $f_R=0.01$ with the other parameters as shown in Table \ref{tab:parameters}.
}
\label{fig:T:geq-1}
\end{figure}

Unlike neutral or strong competition, weak competition ($\alpha < 1$, $\beta < 1$) permits coexistence equilibria regardless of initial states (Figure~\ref{fig:T:leq_1}A), while continuous treatment induces resistant dominance through Darwinian selection (Figure~\ref{fig:T:leq_1}B). Similar to neutral competition, weak competition exhibits three distinct therapeutic response scenarios (Figure~\ref{fig:T:leq_1}C), necessitating strategic treatment selection.

Given clinical uncertainties in weak competition systems, comprehensive parameter-space exploration proves unnecessary and inefficient. The analysis therefore focuses on measurable initial tumor burden ($n_0$) to evaluate adaptive therapy superiority probability under clinical uncertainty, following previous subsections (Figure \ref{fig:T:equal-3}D1-D2). Less accessible parameters including resistant cell fraction ($f_R$), thresholds ($C_\mathrm{TH0}$), and growth rates ($r_S$, $r_R$) are randomized (Figure~\ref{fig:T:leq_1}D0-D8). Analysis reveals enhanced adaptive therapy superiority with decreasing $\alpha$ and increasing $\beta$ (Figure~\ref{fig:T:leq_1}D1-D8), consistent with its premise of leveraging sensitive cells to suppress resistant populations. This relationship weakens when resistant cells evade suppression.

Notably, large initial tumor burdens in strong competition systems reduce adaptive therapy selection probability (Figure~\ref{fig:T:leq_1} upper-right quadrant), primarily due to: 1) bistability-mediated tumor burden homeostasis across treatments negating adaptive advantages, and 2) extended progression timelines for large burdens due to progression criteria, diminishing therapeutic differentiation. Surgical intervention may become preferable when tumor burdens approach carrying capacity. Crucially, adaptive therapy selection probability increases with initial tumor burden except at extreme values (Figure~\ref{fig:T:leq_1}D0), paralleling neutral competition observations (Figures~\ref{fig:T:equal-1}E and \ref{fig:T:equal-3}D1-D2). These findings confirm tumor burden's primacy as the clinically actionable determinant while suggesting benefits from precise competition-strength measurements.

\begin{figure}[phtb]
\centering
\includegraphics[width=\textwidth]{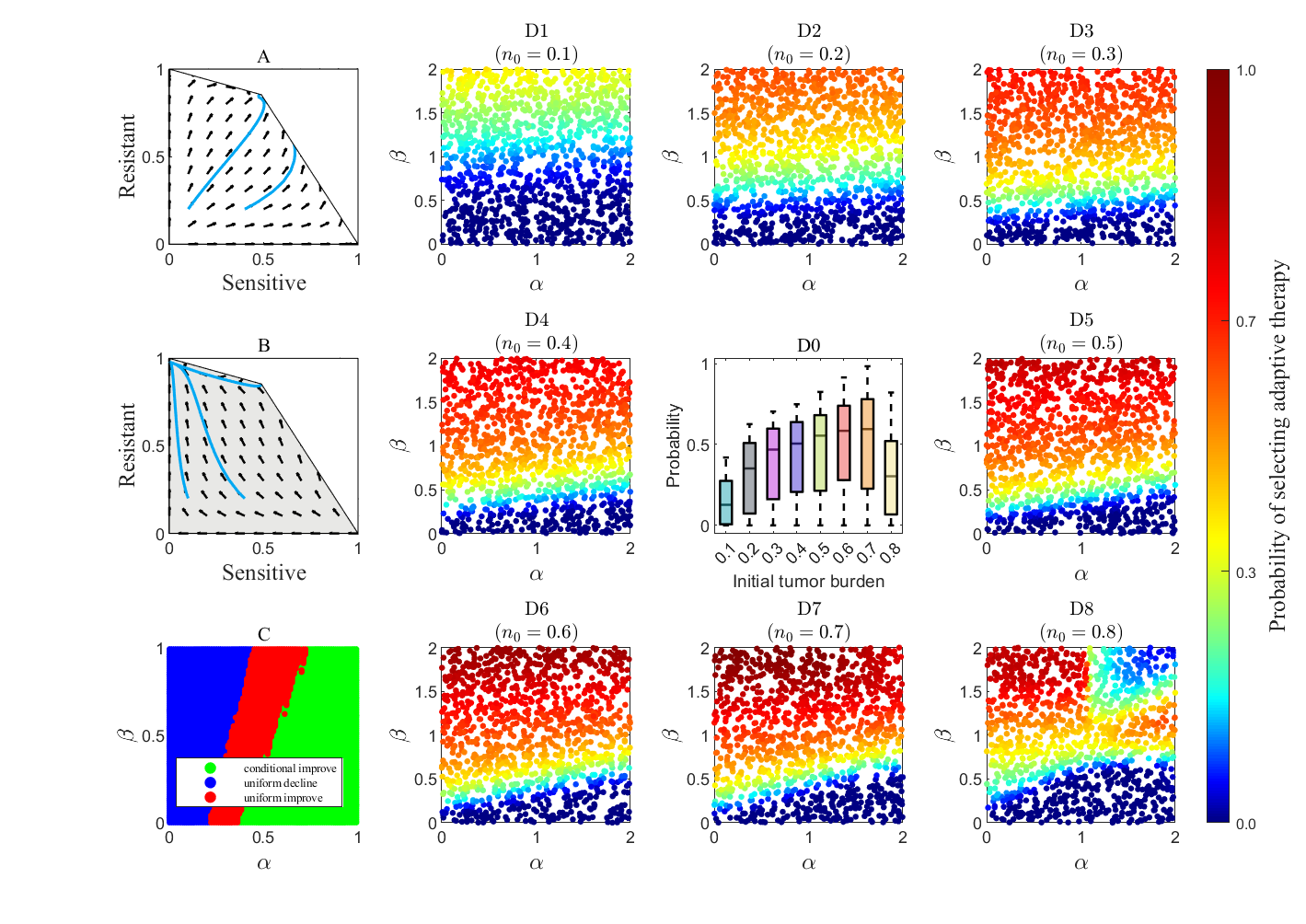}
\caption{
Effectiveness of adaptive therapy under conditions of non-neutral competition.
\textbf{(A)} Phase diagram and vector field of the dynamics of sensitive and resistant cells in the absence of treatment under weak competition. The blank region represents scenarios where the effective tumor burden exceeds the carrying capacity ($S+\alpha R>K$ or $\beta S+R>K$). The blue lines represent tumor evolution trajectories with different initial tumor states.  
\textbf{(B)} Phase diagram and vector field of the dynamics with treatment. The blue lines show that tumor evolution trajectories all move toward a resistant-dominant state. The other elements in the figure are similar to those in panel (A).  
\textbf{(C)} Dependence of the three scenarios on the competition coefficients, with the initial tumor burden $n_0=0.75$, the fraction of resistant cells $f_R=0.01$, and other parameters set as in Table \ref{tab:parameters}.  
\textbf{(D1-D8)} Dependence of the probability of selecting adaptive therapy on competition coefficients, under different initial tumor burdens. Here, the selection of adaptive therapy refers to the condition where the TTP of adaptive therapy is better than that of other strategies. The fraction of resistant cells $f_R$, cell division rates $r_S$ and $r_R$, and treatment-holiday thresholds $C_\mathrm{TH0}$ are randomly sampled to represent various individual conditions. The sample size is $1200$ for each pair of competition coefficient values $\alpha$ and $\beta$.  
\textbf{(D0)} Overall dependence of the probability of selecting adaptive therapy on the initial tumor burden. Boxplots show the aggregated probabilities calculated from panels (D1)-(D8). \textbf{Note} that panels (D0)-(D8) include both weak and strong competition conditions.
}
\label{fig:T:leq_1}
\end{figure}

\subsection{The nonexistence of practical optimal therapeutic strategies}
As shown in Figure \ref{fig:T:equal-1}C3, TTP is positively correlated with the number of treatment holidays. Motivated by this observation, this subsection presents theoretical analysis of the optimal control problem that maximizes TTP through optimization of the dose curve $D(t)$, as formalized in equations \eqref{eq:objective}-\eqref{eq:path constraint}:
\begin{equation}\label{eq:objective}
    \sup_{D(\cdot)} J[D(\cdot)] = \int_{0}^{T}  1  \ \mathrm{d}t = T
\end{equation}
subject to
\begin{align}
    \dfrac{\mathrm{d}S}{\mathrm{d}t} &= r_{S} \left(1 - \dfrac{S + \alpha R}{K}\right) \left(1 - \gamma \dfrac{D(t)}{D_0}\right) S - d_{S}S, \label{eq:dSdt} \\
    \dfrac{\mathrm{d}R}{\mathrm{d}t} &= r_{R} \left(1 - \dfrac{\beta S + R}{K}\right) R - d_{R}R, \label{eq:dRdt} \\
    N(t) &= S(t) + R(t), \label{eq:Ndef} \\
    S(0) &= s_0K, \ R(0) = r_0K, \label{eq:initial} \\
    N(T) &= (1 + \varepsilon) n_0 K, \quad \text{where} \ n_0 = s_0 + r_0, \label{eq:terminal} \\
    N(t) &< (1 + \varepsilon) n_0 K, \quad \text{for} \ t \in [0, T), \label{eq:path constraint}
\end{align}
with parameters satisfying the following constraints \eqref{eq:param1}-\eqref{eq:param4}.
\begin{align}
    &\alpha, \beta, K, r_S, r_R, d_S, d_R > 0, \label{eq:param1} \\
    &\gamma > 1, \quad 0 < \varepsilon \leq 0.5, \label{eq:param2} \\
    &0 \leq D_{\mathrm{min}} \leq D(t) \leq D_{\mathrm{max}} \leq D_0, \label{eq:param3} \\
    &s_0 + \alpha r_0 < 1, \quad \beta s_0 + r_0 < 1. \label{eq:param4}
\end{align}

This optimal control problem aims to maximize Time-to-Progression (TTP) through an optimized dosing schedule, with tumor growth governed by Lotka-Volterra dynamics. For this formulation, Theorem \ref{thm:0:1:full} establishes the nonexistence of practically implementable optimal controls under the specified conditions.

\begin{theorem}\label{thm:0:1:full}
    Consider the optimal control problem defined by \eqref{eq:objective}-\eqref{eq:param4}, where $D^*(t)$ denotes an optimal control maximizing $J[D(\cdot)]$. Assume that $S'(0)>0$ and $R'(0)>0$ when $D(t)\equiv D_{\mathrm{min}}=0$ and $D^*(t)$ is an optimal control. Then any optimal control $D^*(t)$ maximizing $J[D(\cdot)]$ necessarily has infinitely many points of discontinuity if either of the following conditions is satisfied.
    \begin{enumerate}
        \item $\alpha = \beta = 1$, \quad or
        \item $1 \leq \beta \leq \alpha$ \quad and \quad $n_0 < \dfrac{1}{(1+\varepsilon)(1+\beta)}$.
    \end{enumerate}
\end{theorem}

The conditions $S'(0) > 0$ and $R'(0) > 0$ under $D(t) \equiv D_{\mathrm{min}} = 0$ are biologically reasonable, as they indicate that the tumor burden remains below the microenvironment's carrying capacity, enabling tumor growth in the absence of therapy. Furthermore, Theorem \ref{thm:0:1:full} combined with Pontryagin's Maximum Principle implies that any optimal control (if it exists) must undergo infinitely many switches between maximum and minimum dosing levels, rendering it clinically impractical.

Before proving the theorem, some typical calculations for the optimal control problem and lemmas are needed.

Firstly, the necessary condition for the optimal control problem should be calculated using Pontryagin's Maximum Principle. The Hamiltonian $\mathcal{H}$ is defined as equation \eqref{eq:Hdef}.
\begin{equation}
    \begin{array}{l}
        \mathcal{H}(S,R,\lambda_S,\lambda_R,D) = 1 + \lambda_S \dfrac{\mathrm{d}S}{\mathrm{d}t} + \lambda_R \dfrac{\mathrm{d}R}{\mathrm{d}t}.   \\[0.3cm]
        = 1 + \lambda_S \left[ r_{S} \left(1 - \dfrac{S + \alpha R}{K}\right) \left(1 - \gamma \dfrac{D}{D_0}\right) S - d_{S}S \right] + \lambda_R \left[ r_{R} \left(1 - \dfrac{\beta S + R}{K}\right) R - d_{R}R \right].
    \end{array}\label{eq:Hdef}
\end{equation}
Thus, the equations for the costates are given as following eqautions \eqref{eq:dlambdadtS} and \eqref{eq:dlambdadtR}.
\begin{align}
\dfrac{\mathrm{d}\lambda_S}{\mathrm{d}t} &= -\dfrac{\partial \mathcal{H}}{\partial S} \nonumber\\
&= \lambda_S \left[ d_S - r_S \left(1 - \dfrac{2S + \alpha R}{K}\right) \left(1 - \gamma \dfrac{D(t)}{D_0}\right)   \right] + \lambda_R  r_R \dfrac{\beta R}{K}  \label{eq:dlambdadtS} \\
\dfrac{\mathrm{d}\lambda_R}{\mathrm{d}t} &= -\dfrac{\partial \mathcal{H}}{\partial R}, \nonumber\\
&= \lambda_S  r_S \dfrac{\alpha S}{K} \left(1 - \gamma \dfrac{D(t)}{D_0}\right)  + \lambda_R \left[ d_R - r_R \left(1 - \dfrac{\beta S + 2R}{K}\right)  \right]. \label{eq:dlambdadtR}
\end{align}
with the transversality conditions as boundary conditions (equations \eqref{eq:transv1}-\eqref{eq:transv3}).
\begin{eqnarray}
    \lambda_S(T) &=& \nu, \label{eq:transv1}\\
    \lambda_R(T) &=& \nu, \label{eq:transv2}
\end{eqnarray}
where $\nu$ is the Lagrange multiplier associated with the terminal constraint $N(T) = (1 + \varepsilon) n_0 K$, obtained from the terminal condition of the Hamiltonian as given by equation \eqref{eq:transv3}.
\begin{equation*}
    \mathcal{H}(S(T),R(T),\lambda_S(T),\lambda_R(T),D(T)) = 0.
\end{equation*}
So the $\nu$ should satisfies the equation
\begin{equation}
    \mathcal{H}(S(T),R(T),\nu,\nu,D(T))=0.\label{eq:transv3}
\end{equation}
According to Pontryagin's Maximum Principle, since the Hamiltonian $\mathcal{H}$ is linear in $D(t)$ (Equation \eqref{eq:Hdef} and $D^*(t)$ subjects to \eqref{eq:param3}, hence, $D^*(t)$ is determined by the sign of the switching function $\Phi(t)$ given by the equation \eqref{eq:Dopt}.
\begin{equation}
    D^*(t) = 
    \begin{cases} 
        D_{\max}, & \text{if } \Phi(t) > 0, \\
        D_{\min}, & \text{if } \Phi(t) < 0, \\
        \text{singular}, & \text{if } \Phi(t) = 0 \text{ on } [t_1, t_2],
    \end{cases}\label{eq:Dopt}
\end{equation}
where
\begin{equation}
\Phi(t) = -\lambda_S(t) \cdot r_S \left(1 - \dfrac{S(t) + \alpha R(t)}{K}\right) \gamma \dfrac{S(t)}{D_0}.\label{eq:switching}
\end{equation}

In summary, the equations for the optimal solution $D^*(t)$ of the optimal control problem (if it exists) are summarized in the following Lemma \ref{thm:lemma:necessary-condition}.
\begin{lemma}\label{thm:lemma:necessary-condition}
    If the optimal control problem \eqref{eq:objective}-\eqref{eq:param4} has a solution with optimal control $D^*(t)$, then $D^*(t)$ satisfies the following equations:
    \begin{align}
    \dfrac{\mathrm{d}S}{\mathrm{d}t} &= r_{S} \left(1 - \dfrac{S + \alpha R}{K}\right) \left(1 - \gamma \dfrac{D^*(t)}{D_0}\right) S - d_{S}S, \tag{\eqref{eq:dSdt}} \\
    \dfrac{\mathrm{d}R}{\mathrm{d}t} &= r_{R} \left(1 - \dfrac{\beta S + R}{K}\right) R - d_{R}R, \tag{\eqref{eq:dRdt}} \\
    \dfrac{\mathrm{d}\lambda_S}{\mathrm{d}t} &= \lambda_S \left[ d_S - r_S \left(1 - \dfrac{2S + \alpha R}{K}\right) \left(1 - \gamma \dfrac{D^*(t)}{D_0}\right)   \right] + \lambda_R  r_R \dfrac{\beta R}{K},  \tag{\eqref{eq:dlambdadtS}} \\
    \dfrac{\mathrm{d}\lambda_R}{\mathrm{d}t} &= \lambda_S  r_S \dfrac{\alpha S}{K} \left(1 - \gamma \dfrac{^*(t)}{D_0}\right)  + \lambda_R \left[ d_R - r_R \left(1 - \dfrac{\beta S + 2R}{K}\right)  \right], \tag{\eqref{eq:dlambdadtR}}\\
    &N(t) = S(t) + R(t), \tag{\eqref{eq:Ndef}} \\
    &S(0) = s_0K, \ R(0) = r_0K, \tag{\eqref{eq:initial}} \\
    &N(T) = (1 + \varepsilon) n_0 K, \quad \text{where} \ n_0 = s_0 + r_0, \tag{\eqref{eq:terminal}} \\
    &N(t) < (1 + \varepsilon) n_0 K, \quad \text{for} \ t \in [0, T), \tag{\eqref{eq:path constraint}}\\
    &\lambda_S(T) = \nu, \tag{\eqref{eq:transv1}}\\
    &\lambda_R(T) = \nu, \tag{\eqref{eq:transv2}}\\
    &\mathcal{H}(S(T),R(T),\nu,\nu,D^*(T)) = 0,\tag{\eqref{eq:transv3}}\\
    \Phi(t) &= -\lambda_S(t) \cdot r_S \left(1 - \dfrac{S(t) + \alpha R(t)}{K}\right) \gamma \dfrac{S(t)}{D_0},\tag{\eqref{eq:switching}}\\
    D^*(t) &= 
    \begin{cases} 
        D_{\max}, & \text{if } \Phi(t) > 0, \\
        D_{\min}, & \text{if } \Phi(t) < 0, \\
        \text{singular}, & \text{if } \Phi(t) = 0 \text{ on } [t_1, t_2].
    \end{cases}\tag{\eqref{eq:Dopt}}
\end{align}
\end{lemma}
The Lemma \ref{thm:lemma:necessary-condition} also provides a numerical method based on a forward-backward sweep scheme if the optimal control exists.

For simplicity and considering biological realism, hereafter in this section, we make the following assumptions \eqref{eq:Dmin,barD,initial values assumption:1}-\eqref{eq:Dmin,barD,initial values assumption:2}.
\begin{equation}
    D_{\mathrm{min}}=0,\ 1-\gamma\dfrac{D_{\mathrm{max}}}{D_0}<0,\label{eq:Dmin,barD,initial values assumption:1}
\end{equation}
and
\begin{equation}
    S(0),\ R(0),\ S'(0),\ R'(0)>0 \mbox{ when } D(t)\equiv D_{\mathrm{min}}=0.\label{eq:Dmin,barD,initial values assumption:2}
\end{equation}
Biologically, these assumptions indicate that the tumor burden is within the carrying capacity of the microenvironment, enabling tumor growth. 

In addition, to simplify notation, the following symbols are introduced throughout this section. 1) let $\bar{D}=-\left(1-\gamma\dfrac{D_{\mathrm{max}}}{D_0}\right)>0$ denote the factor in \eqref{eq:dSdt}. 2) let $S^*(t),\ R^*(t), N^*(t)=S^*(t)+R^*(t)$ be the state solution for the optimal control problem \eqref{eq:objective}-\eqref{eq:param4}. And 3) let $T^*$be the optimal terminal time (objective value), i.e., the time point such that $N^*(T^*)=(1+\varepsilon)n_0K$ if the optimal control exists.

Based on the necessary condition, more properties of the optimal control $D^*(t)$ can be obtained as Lemma \ref{thm:lemma:lambda_S(T^*)<0, N(T^*)>0}-\ref{thm:lemma:terminal R(t) S(t):2}.
\begin{lemma}\label{thm:lemma:lambda_S(T^*)<0, N(T^*)>0}
    Under the condition of Theorem \ref{thm:0:1:full}, if there exists an optimal control for the optimal control problem \eqref{eq:dSdt}-\eqref{eq:param4}, the costate corresponding to the $S$ variable must be negative at the terminal point $T^*$, $\lambda_S(T^*)$ must be negative and the derivative $\left.\dfrac{\mathrm{d}N^*}{\mathrm{d}t}\right|_{t=T^*}$ must be positive.
\end{lemma}
\begin{proof}
    Under the condition of $D(t)\equiv D_\mathrm{min}=0$, it is straightforward to verify that $S'(t)>0$ and $R'(t)>0$ for all $t \in [0, T]$ except possibly at isolated points if $S'(0)>0$ and $R'(0)>0$ using a standard argument for ordinary differential equations by contradiction as follows. First, for the case that $S'(t)=R'(t)=0$ it will reach a contradiction by the uniqueness theorem of ordinary differential equations, as the constant solution is also a solution. For the case that $S'(t)=0$ while $R'(t)>0$ (without loss of generality). Assume that $t'=\inf\{t>0|S'(t)\leq0\}$ and $t''=\inf\{t>0|R'(t)\leq0\}$, then because $S(t)$ and $R(t)$ is continuously differentiable and $S'(0)>0,\ R'(0)>0$, we have $t'>0$, $t''>0$. So without loss of generality, we can assume that $t'<t''$. In this case, $S'(t')=0$ and $S'(t)>0, R'(t)>0$ for $t<t'$. And thus, $R'(t')>0$ in some semi-neighbourhood of $t'$, $t\in(t',t'+\delta')$. So, it is impossible that $S'(t)=0$ on any neighbourhood of $t'$ because the right-hand side of \eqref{eq:dSdt} cannot keep zero when $R'(t)>0$ under the assumption that $S(t)\neq 0$.
    
    According to \eqref{eq:Dopt}, as $S'(t)>0$ without drugs ($D(t)\equiv0$), it follows that 
    $$\left(1-\dfrac{S(t)+\alpha R(t)}{K}\right)>\dfrac{d_S}{r_S} S(t)>0.$$
    So it is impossible that the case $\Phi(t)\equiv0$ on some interval happens. Thus, the optimal control (if it exists) can be written as \eqref{eq:Dopt:lambda} as follows.
    \begin{equation}\label{eq:Dopt:lambda}
    D^*(t) = \begin{cases} 
    D_{\max}, & \text{if } \lambda_S(t) < 0 \\
    D_{\min}, & \text{if } \lambda_S(t) > 0.
    \end{cases}
    \end{equation}
    
    Therefore, by \eqref{eq:transv1}-\eqref{eq:transv3}, the boundary value of the costate $\lambda_S(t)$ must satisfy
    \begin{equation}
        \lambda_S(T^*)=-\dfrac{1}{\left.\dfrac{\mathrm{dN^*}}{{\mathrm{d}t}}\right|_{t=T^*}}<0. \label{eq:lambda_S(T)<0}
    \end{equation}
    And this completes the proof.
\end{proof}
The Lemma \ref{thm:lemma:lambda_S(T^*)<0, N(T^*)>0} is consistent with the fact that at the time of progression, the total tumor burden is growing and progression always occurs during drug administration. Otherwise, intuitively, drug administration can always slow down the growth of tumor within a small time interval as sensitive cells are being eliminated.

\begin{lemma}\label{thm:lemma:terminal R(t) S(t):1}
    Assume the optimal control problem has a solution that is piecewise constant with finitely many discontinuities. If $\alpha=\beta=1$, there exists $T_0<T^*$ such that $r_R R^*(t) > \bar{D}  r_S S^*(t)$ holds for every $t\in [T_0,T^*]$.
\end{lemma}
\begin{proof}
    For simplicity, the asterisk in the superscript are omitted in this proof. According to Lemma \ref{thm:lemma:lambda_S(T^*)<0, N(T^*)>0}, in some neighbourhood of $T^*$, the following holds.
    \begin{eqnarray}
        \dfrac{\mathrm{d}N}{\mathrm{d}t}&=& -r_{S} \left(1 - \dfrac{S + R}{K}\right) \bar{D} S - d_{S}S+ r_{R} \left(1 - \dfrac{ S + R}{K}\right) R - d_{R}R\nonumber\\
        &=&\left(1-\dfrac{S+R}{K}\right)(-r_S\bar{D}S+r_RR)-d_SS-d_RR>0.
    \end{eqnarray}
    So $\left(1-\dfrac{S+R}{K}\right)(-r_S\bar{D}S+r_RR)>d_SS+d_RR>0$ and therefore
    \begin{equation}
        r_RR(t)>\bar{D}r_SS(t), \ \mbox{ for } t\in [T_0,T^*].
    \end{equation}
\end{proof}
\begin{lemma}\label{thm:lemma:terminal R(t) S(t):2}
    Assume that $1\leq\beta\leq\alpha$ and $S(0)+R(0)<\dfrac{K}{(1+\varepsilon)(1+\beta)}$. If the optimal control problems \eqref{eq:objective}-\eqref{eq:param4} has a piecewise constant optimal control according to PMP, there exists $T_0'<T^*$ such that $r_R\beta R^*(t) > \bar{D}  r_S S^*(t)$ for every $t\in [T_0', T^*]$.
\end{lemma}
\begin{proof}
    Similar to the proof of Lemma \ref{thm:lemma:terminal R(t) S(t):1}, the asterisks are also omitted in this proof. And, in a neighbourhood of $T^*$, the following holds.
    \begin{eqnarray}
        \dfrac{\mathrm{d}N}{\mathrm{d}t}&=& r_{S} \left(1 - \dfrac{S + \alpha R}{K}\right) \left(-\bar{D}\right) S - d_{S}S+ r_{R} \left(1 - \dfrac{ \beta S + R}{K}\right) R - d_{R}R>0.
    \end{eqnarray}
    And thus,
    \begin{eqnarray}
        \left(1 - \dfrac{S + \alpha R}{K}\right)\left(-\bar{D}r_SS+r_R\beta R\right)+r_RR\left(\dfrac{(\beta-\alpha)S}{K}+(1-\beta)\dfrac{K-(\beta+1)R}{K}\right)-d_Ss-d_RR>0,\nonumber\\
    \end{eqnarray}
    Moreover, since $n_0K=S(0)+R(0)<\dfrac{K}{(1+\varepsilon)(1+\beta)}$, the threshold of progression is thus $(1+\varepsilon)n_0K<\dfrac{K}{1+\beta}$. So the final state of $R(t)$ should satisfies $S(T^*)<N(T^*)<\dfrac{K}{1+\beta}$. Therefore, $\dfrac{(\beta-\alpha)S}{K}+(1-\beta)\dfrac{K-(\beta+1)R}{K}<0$. Also because that $1<\beta<\alpha$, so we have that $-\bar{D}r_SS+r_R\beta R>0$ which, according to the simplification of the notations, implies that $r_R\beta R^*(T^*)>\bar{D}r_SS^*(T^*)$.
    
    And the result of the lemma is followed by a common argument of calculus.
\end{proof}

Now the lemmas are sufficient to prove the main Theorem \ref{thm:0:1:full}.
\begin{proof}[Proof of Theorem \ref{thm:0:1:full}]
    By contradiction, assume the optimal control $D^*(t)$ has only finitely many discontinuities. Then, by Pontryagin's Maximum Principle (PMP), $D^*(t)$ is a piecewise constant function whose range has only two values. Let $T_1$ denote the last discontinuity point of the optimal control. So we have that $D^*(t)=D_\mathrm{max}$ for $t\in[T_1,T^*]$ and moreover, from Lemma \ref{thm:lemma:terminal R(t) S(t):1} and Lemma \ref{thm:lemma:terminal R(t) S(t):2}, we can choose $T_2\in[T_1,T^*)$ such that $N'(t)>0$ and $r_R\beta R(t)>\bar{D}r_SS(t)$ hold for all $t\in [T_2,T^*]$. 

    Now we are going to prove that there exists $\Delta t>0 $ such that the state solution has a better termination time than $T^*$ with $D_{\Delta t}(t)$ defined as the following \eqref{eq:D1:for theorem}.
    \begin{equation}
        D_{\Delta t}(t)=\begin{cases}
            D_{\mathrm{min}},& t\in[T^*-2\Delta t,T^*-\Delta t],\\
            D^*(t), & \mbox{ otherwise}.
        \end{cases}\label{eq:D1:for theorem}
    \end{equation}
    Let $S_{\Delta t}(t),\ R_{\Delta t}(t),\ N_{\Delta t}(t)=S_{\Delta t}(t)+R_{\Delta t}(t)$ denote the solution with this $D_{\Delta t}(t)$. Then $N^*(t)$ and $N_{\Delta t}(t)$ obey the same differential equations for all $t$ except $t\in[T^*-2\Delta t,T^*-\Delta t]$.
    
    Moreover, by denoting $t_i=T^*-(2-i)\Delta t$, $i=0,1,2$ and $S^*_0=S^*(t_0), \ R^*_0=R^*(t_0)$, then we have equations \eqref{eq:ode:Euler:1}-\eqref{eq:ode:Euler:8}.
    \begin{align}
        S^*(t_1) &= S^*_0 + \left[ r_S \left(1 - \dfrac{S^*_0 + \alpha R^*_0}{K}\right) (-\bar{D}) S^*_0 - d_S S^*_0 \right] \Delta t + o(\Delta t), \label{eq:ode:Euler:1} \\
        R^*(t_1) &= R^*_0 + \left[ r_R \left(1 - \dfrac{\beta S^*_0 + R^*_0}{K}\right) R^*_0 - d_R R^*_0 \right] \Delta t + o(\Delta t), \label{eq:ode:Euler:2} \\
        S^*(t_2) &= S^*(t_1) + \left[ r_S \left(1 - \dfrac{S^*(t_1) + \alpha R^*(t_1)}{K}\right) (-\bar{D}) S^*(t_1) - d_S S^*(t_1) \right] \Delta t + o(\Delta t), \label{eq:ode:Euler:3} \\
        R^*(t_2) &= R^*(t_1) + \left[ r_R \left(1 - \dfrac{\beta S^*(t_1) + R^*(t_1)}{K}\right) R^*(t_1) - d_R R^*(t_1) \right] \Delta t + o(\Delta t), \label{eq:ode:Euler:4} \\[0.3cm]
        S_{\Delta t}(t_1) &= S^*_0 + \left[ r_S \left(1 - \dfrac{S^*_0 + \alpha R^*_0}{K}\right) S^*_0 - d_S S^*_0 \right] \Delta t + o(\Delta t), \label{eq:ode:Euler:5} \\
        R_{\Delta t}(t_1) &= R^*_0 + \left[ r_R \left(1 - \dfrac{\beta S^*_0 + R^*_0}{K}\right) R^*_0 - d_R R^*_0 \right] \Delta t + o(\Delta t), \label{eq:ode:Euler:6} \\
        S_{\Delta t}(t_2) &= S_{\Delta t}(t_1) + \left[ r_S \left(1 - \dfrac{S_{\Delta t}(t_1) + \alpha R_{\Delta t}(t_1)}{K}\right) (-\bar{D}) S_{\Delta t}(t_1) - d_S S_{\Delta t}(t_1) \right] \Delta t + o(\Delta t), \label{eq:ode:Euler:7} \\
        R_{\Delta t}(t_2) &= R_{\Delta t}(t_1) + \left[ r_R \left(1 - \dfrac{\beta S_{\Delta t}(t_1) + R_{\Delta t}(t_1)}{K}\right) R_{\Delta t}(t_1) - d_R R_{\Delta t}(t_1) \right] \Delta t + o(\Delta t). \label{eq:ode:Euler:8}
    \end{align}
    And therefore, we have that
    \begin{equation}
        \displaystyle
        \begin{array}{l}
            N^*(T^*)-N_{\Delta t}(T^*)\\
            =S^*(t_2)+R^*(t_2)-(S_{\Delta t}(t_2)+R_{\Delta t}(t_2)\\[0.3cm]
            =(1 + \bar{D}) r_S S^*_0 \left(1 - \dfrac{S^*_0 + \alpha R^*_0}{K}\right) \left[d_S  + \bar{D}  r_S\left(1 -  \dfrac{S^*_0 +    \alpha R^*_0}{K}\right) +\dfrac{-\bar{D}r_S S^*_0+r_R\beta R^*_0}{K}\right] \left(\Delta t\right)^2+o\left(\left(\Delta t\right)^2\right).
        \end{array}
    \end{equation}
    Hence, for the case 1, $\alpha=\beta=1$, it has already been obtained that $1-\dfrac{S^*(t)+R^*(t)}{K}>0$ and $\bar{D}  r_S S^*_0< r_R R^*_0$ for $\Delta t <\dfrac{T^*-T_0}{2}$ according to Lemma \ref{thm:lemma:terminal R(t) S(t):1}. And since the $o(\Delta t)$ only depends on the properties of the (assumed) optimal solution $N^*(t)$,  there exists a $\delta'<\dfrac{T^*-T_0}{2}$ such that for $\Delta t\in (0,\delta')$, and therefore, $N^*(T^*)-N_{\Delta t}(T^*)>0$, implying that $N(t)$ reaches $(1+\varepsilon)n_0K$ at some $t>T^*$, contradicting the assumption that $T^*$ is the optimal objective value.

    For the case 2, $1\leq\beta\leq\alpha$ and $n_0<\dfrac{1}{(1+\varepsilon)(1+\beta)}$, it follows that $N^*(T^*)-N_{\Delta t}(T^*)>0$ by Lemma \ref{thm:lemma:terminal R(t) S(t):2}, which also contradicts with the assumption that $N^*(t)$ is the optimal state solution.

    So the proof has been completed.
\end{proof}

\section{DISCUSSION}
Adaptive therapy crucially leverages competitive dynamics to maintain therapeutic-sensitive cell reservoirs that suppress resistant clones, strategically delaying drug resistance evolution while minimizing tumor progression \cite{Gatenby2009}. By preserving sensitive cells through threshold-modulated treatment holidays, adaptive therapy sustains ecological suppression of resistant populations—an advantage over traditional maximum tolerated dose regimens, which inadvertently accelerate resistance \cite{Enriquez-Navas2016}. The clinical success of adaptive therapy hinges on its paradigm-shifting principle: harnessing intra-tumor competition to transform therapy-sensitive cells into permanent controllers of resistant cell expansion \cite{West2020a, Zhang2022}. However, the rule-of-thumb thresholds used in clinical practice may require further research and optimization as they might not fully account for the complex, dynamic interactions between sensitive and resistant cell populations, and may vary across different tumor types and patient profiles \cite{McGehee2024.02.19.580916, Wang2024c, Tan2024,Strobl2023.03.22.533721,Strobl2023b,Liu2024d}.

This study employs a Lotka-Volterra model to investigate how treatment-holiday thresholds influence adaptive therapy outcomes in tumors composed of drug-sensitive and resistant cells. By simulating tumor dynamics under varying thresholds, initial tumor burdens, and competition coefficients, the analysis identifies three distinct therapeutic scenarios: \textbf{uniform-improve} (where AT consistently outperforms MTD regardless of threshold selection), \textbf{conditional-improve} (AT efficacy depends on threshold calibration), and \textbf{uniform-decline} (MTD superiority). When considering clinical uncertainties in measurement data, generally, the greater the initial tumor burden, the more likely adaptive therapy will outperform other treatment strategies. However, in cases where the tumor burden is excessively high, the advantages of adaptive therapy become irrelevant, and alternative approaches, beyond pharmacological treatment, may be required.

The competitive dynamics between sensitive and resistant cells is also a key factor in tumor evolution under therapeutic stress. In strong competition environments, tumors tend to evolve toward a state dominated by either sensitive or resistant cells, depending on the initial composition. Continuous MTD therapy drives the tumor toward a stable state where the ratio of sensitive to resistant cells remains constant. Thus, the critical factor in treatment outcomes is not the treatment-holiday threshold but rather the impact of the initial tumor state on the final result. In weak competition environments, sensitive and resistant cells can coexist, and adaptive therapy may lead to resistant cell dominance through Darwinian selection. Treatment effectiveness is influenced by the competition coefficients, $\alpha$ and $\beta$, with lower $\alpha$ and higher $\beta$ favoring adaptive therapy. Adaptive therapy prolongs TTP most effectively when sensitive cells maintain a competitive advantage and when initial tumor burden is high, enabling ecological suppression of resistant cells.

This study advances our understanding of adaptive therapy by systematically quantifying the role of treatment-holiday thresholds in modulating tumor evolution under treatment. By integrating ecological competition principles into therapeutic design, the current work provides a mechanistic basis for optimizing intermittent dosing strategies. The findings highlight that adaptive therapy is not universally superior to MTD; its success depends critically on tumor initial state and the treatment-holiday threshold selection. These insights bridge theoretical ecology and clinical oncology, emphasizing the need for dynamic, patient-specific treatment protocols rather than static dosing regimens.  

This work establishes a quantitative relationship between treatment-holiday thresholds and progression-free survival. By categorizing outcomes into three distinct scenarios, the model offers clinicians a potential framework to stratify patients based on measurable parameters such as initial tumor burden and resistant cell fraction. Furthermore, the demonstration that strong competition can stabilize tumor burden indefinitely under specific thresholds has profound implications for managing advanced cancers. These results extend prior work \cite{Gatenby2009, Strobl2021a,McGehee2024.02.19.580916,Wang2024c} by incorporating threshold-driven mechanism in adaptive therapy.

Despite these advances, the study has several limitations. First, the binary classification of tumor cells overlooks the phenotypic plasticity and continuous heterogeneity that are critical drivers of therapeutic resistance \cite{Gupta2011a, Marusyk2020, Brabletz2018}. Second, model parameters such as the competition coefficients ($\alpha$, $\beta$) are based on population-level assumptions rather than molecular data, necessitating further investigation into the microscopic foundations of these assumptions \cite{Uthamacumaran2022, Wang2020}. Third, the Lotka-Volterra model does not account for the spatial structure of tumors \cite{Noble2022} or the microenvironmental constraints that alter tumor dynamics \cite{West2018, Lorenzi2016}. Future studies should integrate single-cell and spatial omics data to mechanistically parameterize competition coefficients and explore hybrid models that combine continuum and spatial interactions. Additionally, clinical validation through adaptive therapy trials, incorporating real-time biomarker monitoring, will be crucial for translating these insights into practice.







\subsection*{DATA AND CODE AVAILABILITY}

\begin{itemize}
    \item This study did not use any real-world data; all data were generated through computational simulations.
    \item All original code has been deposited on GitHub (\verb|https://github.com/zhuge-c/AT-0|).
    \item Any additional information required to reanalyze the data reported in this paper is available from the lead contact upon request.
\end{itemize}

\subsection*{ACKNOWLEDGMENTS}
This work was funded by The National Natural Science Foundation of China (NSFC) via grant 11801020.





\newpage

\bibliography{ref-0}
\bibliographystyle{plain}

\bigskip

\newpage








\newpage
\appendix
\renewcommand{\thefigure}{S\arabic{figure}}
\setcounter{figure}{0}

\newpage

\section{Details of the dependence of TTP on various parameters in neutral competition condition}
By varying the initial proportion of resistant cells and the initial total cell count under different treatment-holiday thresholds (Figure \ref{fig:S:equal-2}A1-A3), we observed that when the initial total cell count is small and the initial proportion of resistant cells is high, the TTP is short. In this case, although the tumor size is small, it contains a large proportion of resistant cells, making it difficult for sensitive cells to suppress resistant cells, and thus cancer becomes harder to treat. As the initial total cell count increases and the initial proportion of resistant cells decreases, the TTP gradually increases. At this stage, although the number of resistant cells is low, the tumor is larger, requiring multiple treatment cycles in adaptive therapy to suppress resistant cells. Overall, larger adaptive thresholds, higher initial total cell counts, and lower initial proportions of resistant cells result in better treatment outcomes.

By varying the treatment-holiday threshold for dose-skipping and the initial proportion of resistant cells under different initial total cell proportions (Figure \ref{fig:S:equal-2}B1-B3), we found that lower treatment-holiday thresholds and higher initial proportions of resistant cells result in shorter TTPs. This is because tumors containing a high proportion of resistant cells make it difficult for adaptive therapy to reduce the total cell count below the threshold. As the adaptive threshold increases and the initial proportion of resistant cells decreases, the TTP gradually increases. When the adaptive threshold is high and the initial proportion of resistant cells is low, the TTP is maximized. At this stage, the tumor contains a significant number of sensitive cells, and a high threshold adaptive therapy can increase the number of treatment cycles, fully utilizing sensitive cells to suppress resistant cell growth. Overall, higher initial total cell counts are associated with longer TTPs.

By varying the treatment-holiday threshold for dose modulation and the initial total cell count under different initial proportions of resistant cells (Figure \ref{fig:S:equal-2}C1-C3), we observed that lower treatment-holiday thresholds and lower initial total cell counts result in shorter TTPs. This occurs because, despite the smaller tumor size, the low adaptive threshold leads to prolonged treatment cycles, causing significant damage to sensitive cells and reducing their ability to suppress resistant cells, resulting in shorter TTPs. As the adaptive threshold increases and the initial total cell count rises, the TTP gradually increases. When the adaptive threshold and the initial total cell count are both high, the TTP is maximized. At this stage, although the tumor size is large, the shorter treatment cycles allow sensitive cells sufficient time to recover and suppress resistant cell growth, resulting in longer TTPs. Overall, lower initial proportions of resistant cells are associated with longer TTPs.

Under the same conditions, simulations that varied the growth rates of sensitive and resistant cells revealed that when the ratio of the growth rate of sensitive cells to resistant cells ($r_S/r_R$) is greater than $1$, indicating that the growth rate of sensitive cells exceeds that of resistant cells, the sensitive cells exhibit stronger competitiveness, resulting in longer TTPs. Conversely, when the ratio is less than $1$, indicating that the growth rate of resistant cells exceeds that of sensitive cells, the TTP is shorter, and the higher the growth rate of resistant cells, the shorter the TTP. This relationship is illustrated in Figures \ref{fig:S:equal:supplement-1} through \ref{fig:S:equal:supplement-5}.

\section{Details of comparison of adaptive therapy with other strategies in the weak competition condition}
To investigate the probability of selecting adaptive therapy under varying parameters, we randomly varied five key model parameters from Table \ref{tab:parameters}, creating a cohort of $1200$ virtual patients. Simulations across different initial tumor burdens revealed that as the initial cell count increases, the probability of selecting adaptive therapy also increases (Main Text Figure \ref{fig:T:leq_1}D1-D7). However, when the competition coefficients satisfy $\alpha>1$ and $\beta>1$, and the initial cell count approaches the carrying capacity, the probability of selecting adaptive therapy decreases (Main Text Figure \ref{fig:T:leq_1}D8). This is because, under high competition coefficients, tumor competition becomes intense, and both MTD and adaptive therapy can control tumor size indefinitely, leading to a lower probability of selecting adaptive therapy with adaptive therapy (Figure \ref{fig:S:supplement_geq_1}).

In the case of conditional-improve, the differences in treatment outcomes among the four optimal strategies are small (approximately one month; $C_\mathrm{TH0}=0.98$, $(\delta_1, \delta_2, C_\mathrm{TH2}-1, C_\mathrm{TH1}-1)=(0.25, 0.25, 0.05, 0.07)$, $(T_D, T)=(9, 10)$). Since intermittent therapy requires a treatment holiday, the calculation of its optimal cycle revealed a one-day treatment holiday, resulting in treatment outcomes similar to MTD (Figure \ref{fig:S:leq-2}A). Considering cumulative drug toxicity, AT-S demonstrates the best efficiency among the four strategies (Figure \ref{fig:S:leq-2}B).

Subsequently, we varied the initial tumor parameters to analyze the treatment outcomes of the four strategies. When the initial proportion of resistant cells remains constant, and the initial cell count is small, the differences in TTP extension among the strategies are minimal. As the initial cell count increases, AT-S outperforms the other strategies in extending TTP, while AT-M shows the best Relative efficiency (Figure \ref{fig:S:leq-2}C).

When the initial tumor burden remains constant, and the initial proportion of resistant cells is small, AT-S outperforms the other strategies. However, as the initial proportion of resistant cells increases, MTD, AT-S, and IT outperform AT-M, with AT-S and AT-M demonstrating the best Relative efficiency among the four strategies (Figure \ref{fig:S:leq-2}D).

For the uniform-improve scenario, the differences in treatment outcomes among the four optimal strategies are minimal (approximately one month; $C_\mathrm{TH0}=0.98$, $(\delta_1, \delta_2, C_\mathrm{TH2}-1, C_\mathrm{TH1}-1)=(0.25, 0.25, 0.05, 0.09)$, $(T_D, T)=(20, 21)$). Similar to the conditional-improve scenario, the optimal intermittent therapy cycle includes a one-day treatment holiday, resulting in treatment outcomes similar to MTD (Figure \ref{fig:S:leq-3}A). Considering cumulative drug toxicity, AT-S shows the best efficiency among the four strategies (Figure \ref{fig:S:leq-3}B).

When varying the initial total cell count while keeping the initial proportion of resistant cells constant, the differences in TTP extension among the four strategies are minor. MTD and IT outperform AT-S and AT-M in terms of TTP, but AT-S and AT-M have a better relative efficiency than MTD and IT (Figure \ref{fig:S:leq-3}C). When varying the initial proportion of resistant cells while keeping the initial tumor burden constant, MTD and IT outperform AT-S and AT-M in terms of TTP. However, AT-S and AT-M exhibit a better Relative efficiency compared to MTD and IT (Figure \ref{fig:S:leq-3}D).

For the uniform-decline scenario, the differences in TTP among the four optimal strategies are larger (approximately two months; $C_\mathrm{TH0}=0.98$, $(\delta_1, \delta_2, C_\mathrm{TH2}-1, C_\mathrm{TH1}-1)=(0.25, 0.25, 0.05, 0.05)$, $(T_D, T)=(15, 18)$). Among the four strategies, AT-S and AT-M demonstrate the best treatment outcomes (Figure \ref{fig:S:leq-4}A). Regarding cumulative drug toxicity, AT-S achieves the best relative efficiency (Figure \ref{fig:S:leq-4}B).

When varying the initial total cell count while keeping the initial proportion of resistant cells constant, the differences in treatment outcomes are minor for small tumors. For larger tumors, AT-S shows the best treatment outcomes, while AT-M and AT-S have a better Relative efficiency than the other two strategies (Figure \ref{fig:S:leq-4}C). When varying the initial proportion of resistant cells while keeping the initial tumor burden constant, AT-S demonstrates the best treatment outcomes regardless of the proportion of resistant cells. AT-M shows slightly worse treatment outcomes than the other three strategies but exhibits the best Relative efficiency (Figure \ref{fig:S:leq-4}D).
\newpage

\section{Supplementary figures}
\newpage

\begin{figure}[tb]
\centering
\includegraphics[width=\textwidth]{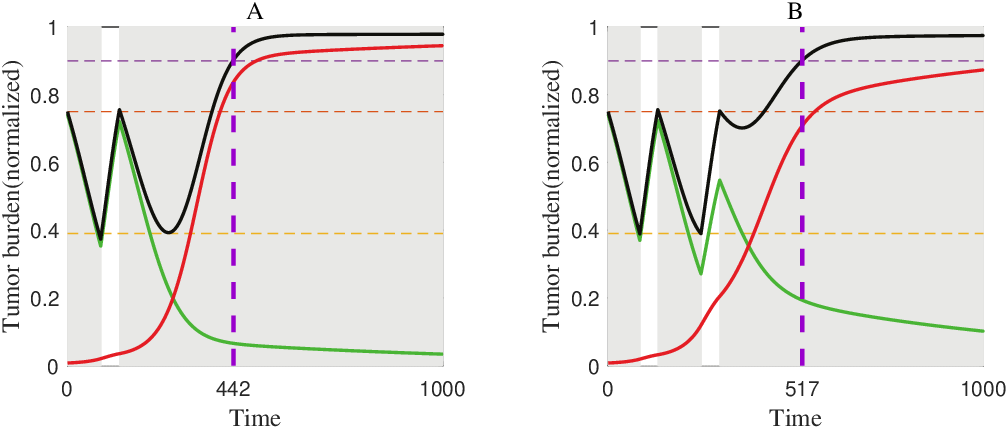}
\caption{Tumor evolution dynamics in neutral competition with different treatment-holiday thresholds, showing the reason of discontinuity in the dependence of TTP on the treatment-holiday thresholds (Figure \ref{fig:T:equal-1}). Purple dashed line: TTP; Red: Resistant cells; Green: Sensitive cells; Black: Total cell count; Shaded area: Treatment implementation.
\textbf{(A)} $n_0=0.75$, $f_R=0.01$, $C_\mathrm{TH0}=0.51$.
\textbf{(B)} $n_0=0.75$, $f_R=0.01$, $C_\mathrm{TH0}=0.52$. Other parameters are taken as those in Table \ref{tab:parameters}.
}
\label{fig:S:supplement_cellcount}
\end{figure}

\begin{figure}[phtb]
\centering
\includegraphics[width=\textwidth]{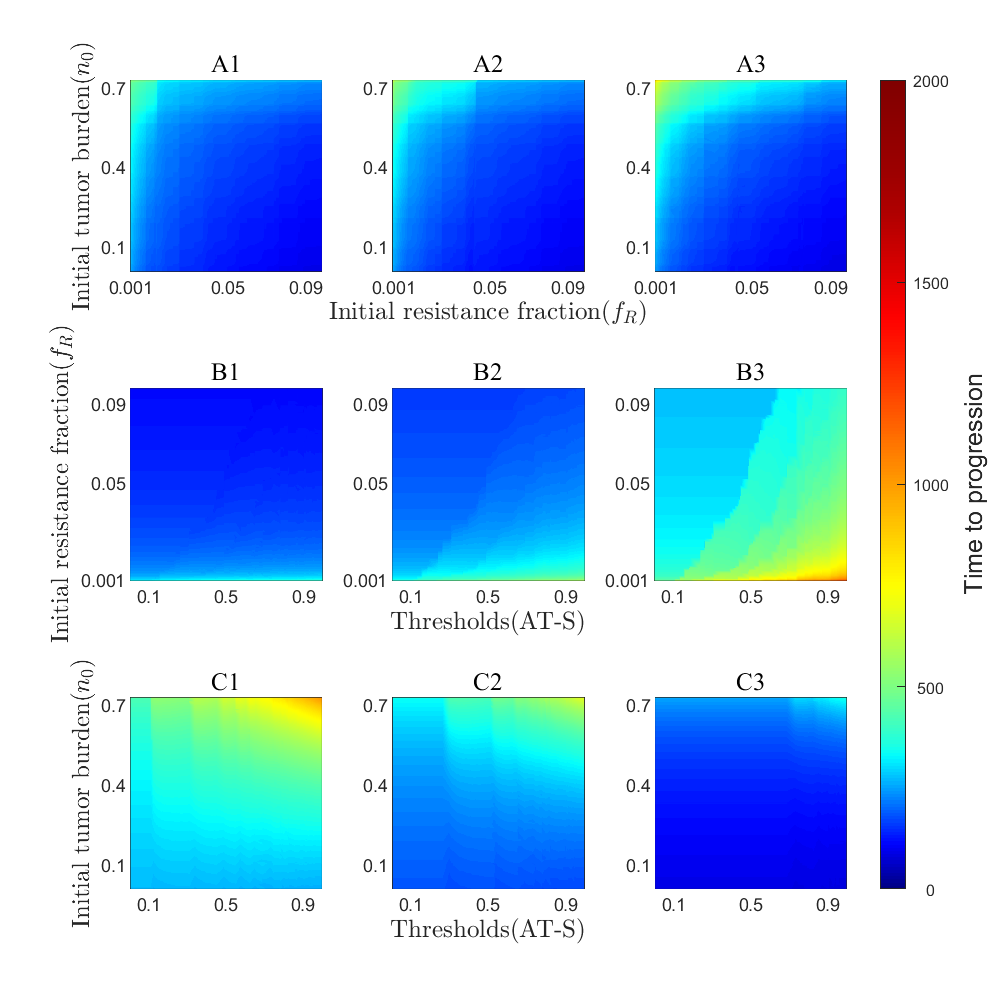}
\caption{The impact of initial tumor state and treatment-holiday thresholds on TTP.
\textbf{(A1)}-\textbf{(A3)} The effect of varying the initial proportion of resistant cells and the initial total cell count on TTP at fixed adaptive thresholds. A1: $C_{TH0} = 0.3$; A2: $C_{TH0} = 0.5$; A3: $C_{TH0} = 0.8$.
\textbf{(B1)}-\textbf{(B3)} The effect of varying treatment-holiday thresholds and the initial proportion of resistant cells on TTP at fixed initial cell counts. B1: $n_{0} = 0.25$; B2: $n_{0} = 0.5$; B3: $n_{0} = 0.75$.
\textbf{(C1)}-\textbf{(C3)} The effect of varying treatment-holiday thresholds and initial cell counts on TTP at fixed proportions of resistant cells. C1: $f_{R} = 0.001$; C2: $f_{R} = 0.01$; C3: $f_{R} = 0.1$.
}
\label{fig:S:equal-2}
\end{figure}

\begin{figure}[tb]
\centering
\includegraphics[width=\textwidth]{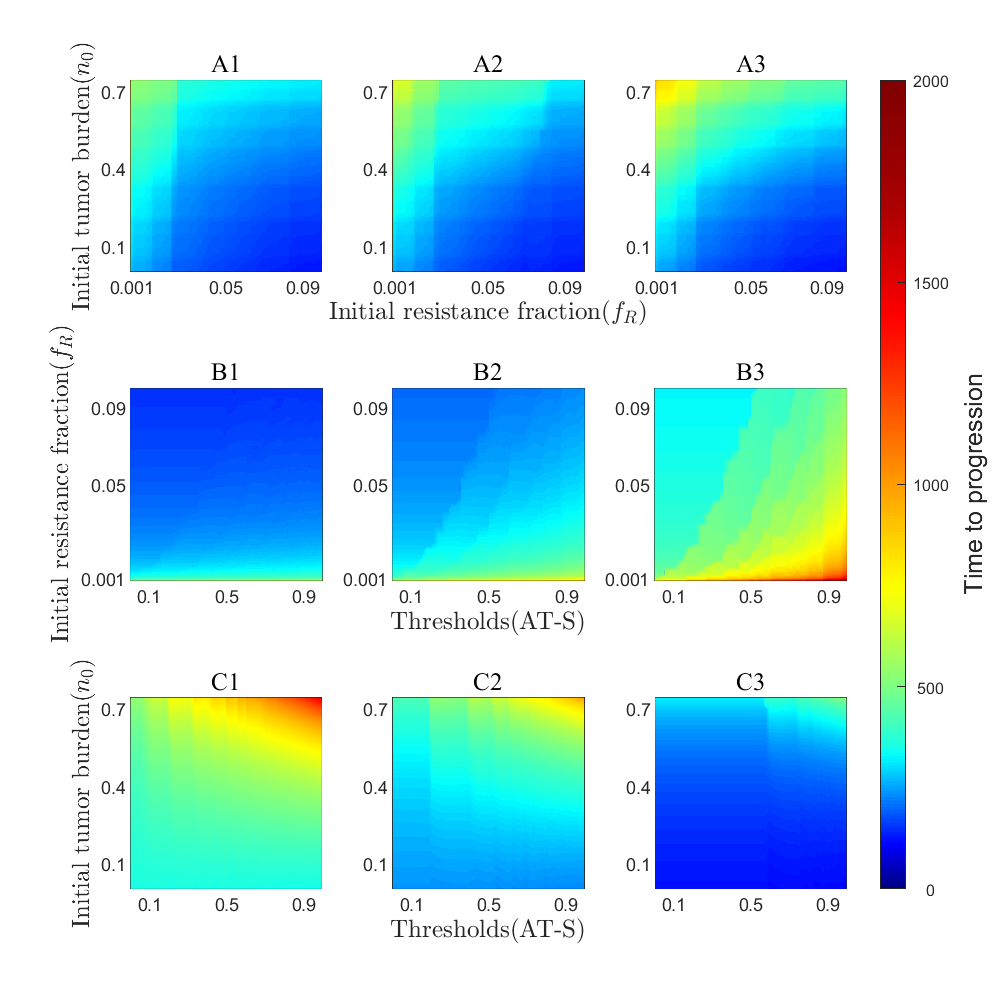}
\caption{The dependence of TTP on the thresholds, initial tumor burden ($n_0$) and initial fraction of resistant cells ($f_R$) in neutral competition with $r_S/r_R=1.8863$, ${r}_{S}=0.0389$.
\textbf{(A1)}-\textbf{(A3)} Effects of varying initial resistance proportions and cell counts on TTP, A1: $C_\mathrm{TH0}=0.3$; A2: $C_\mathrm{TH0}=0.5$; A3: $C_\mathrm{TH0}=0.8$
\textbf{(B1)}-\textbf{(B3)} Effects of varying treatment-holiday thresholds and initial resistance proportions on TTP. B1: $n_0=0.25$; B2: $n_0=0.5$; B3: $n_0=0.75$
\textbf{(C1)}-\textbf{(C3)} Effects of varying treatment-holiday thresholds and initial tumor burden on TTP, C1: $f_{R}=0.001$; C2: $f_{R}=0.01$; C3: $f_{R}=0.1$.
}
\label{fig:S:equal:supplement-1}
\end{figure}

\begin{figure}[tb]
\centering
\includegraphics[width=\textwidth]{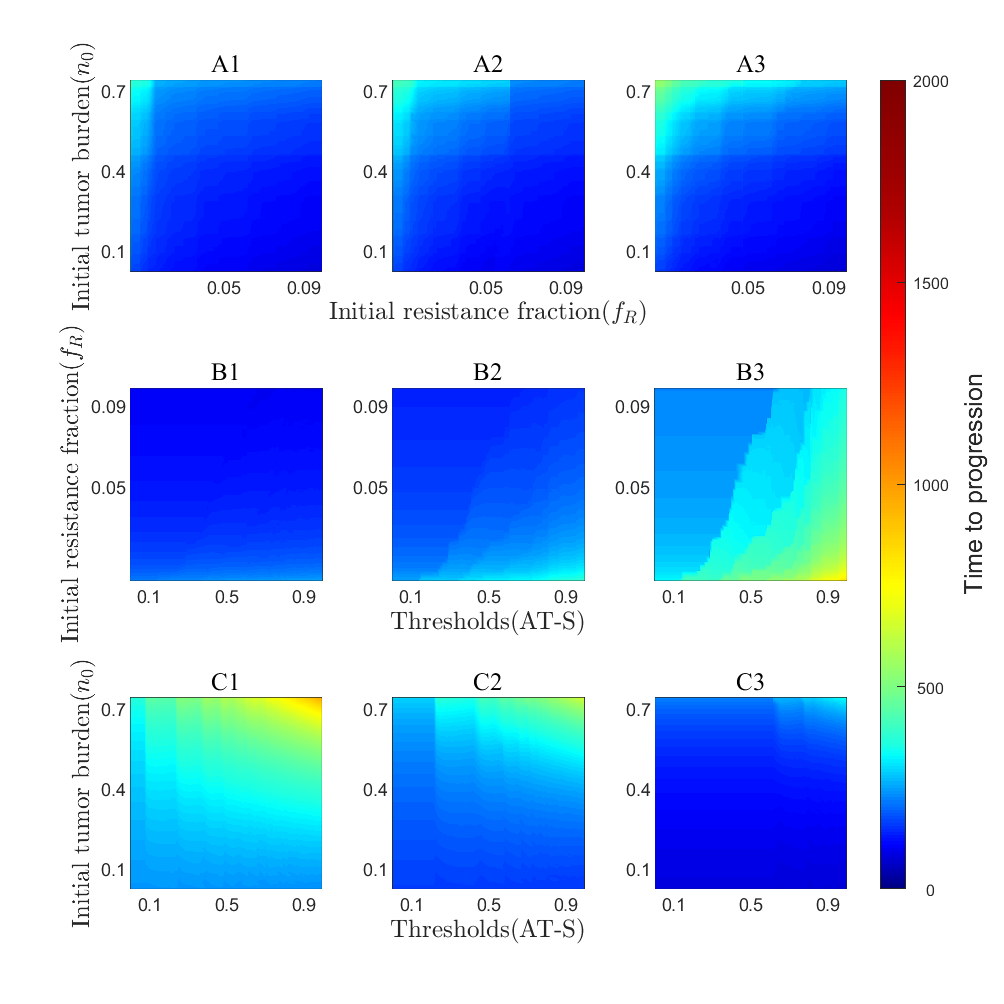}
\caption{The dependence of TTP on the thresholds, initial tumor burden ($n_0$) and initial fraction of resistant cells ($f_R$) in neutral competition with $r_S/r_R=1.6493$, ${r}_{S}=0.0508$.
\textbf{(A1)}-\textbf{(A3)} Effects of varying initial resistance proportions and cell counts on TTP, A1: $C_\mathrm{TH0}=0.3$; A2: $C_\mathrm{TH0}=0.5$; A3: $C_\mathrm{TH0}=0.8$
\textbf{(B1)}-\textbf{(B3)} Effects of varying treatment-holiday thresholds and initial resistance proportions on TTP. B1: $n_0=0.25$; B2: $n_0=0.5$; B3: $n_0=0.75$
\textbf{(C1)}-\textbf{(C3)} Effects of varying treatment-holiday thresholds and initial tumor burden on TTP, C1: $f_{R}=0.001$; C2: $f_{R}=0.01$; C3: $f_{R}=0.1$.
}
\label{fig:S:equal:supplement-2}
\end{figure}

\begin{figure}[tb]
\centering
\includegraphics[width=\textwidth]{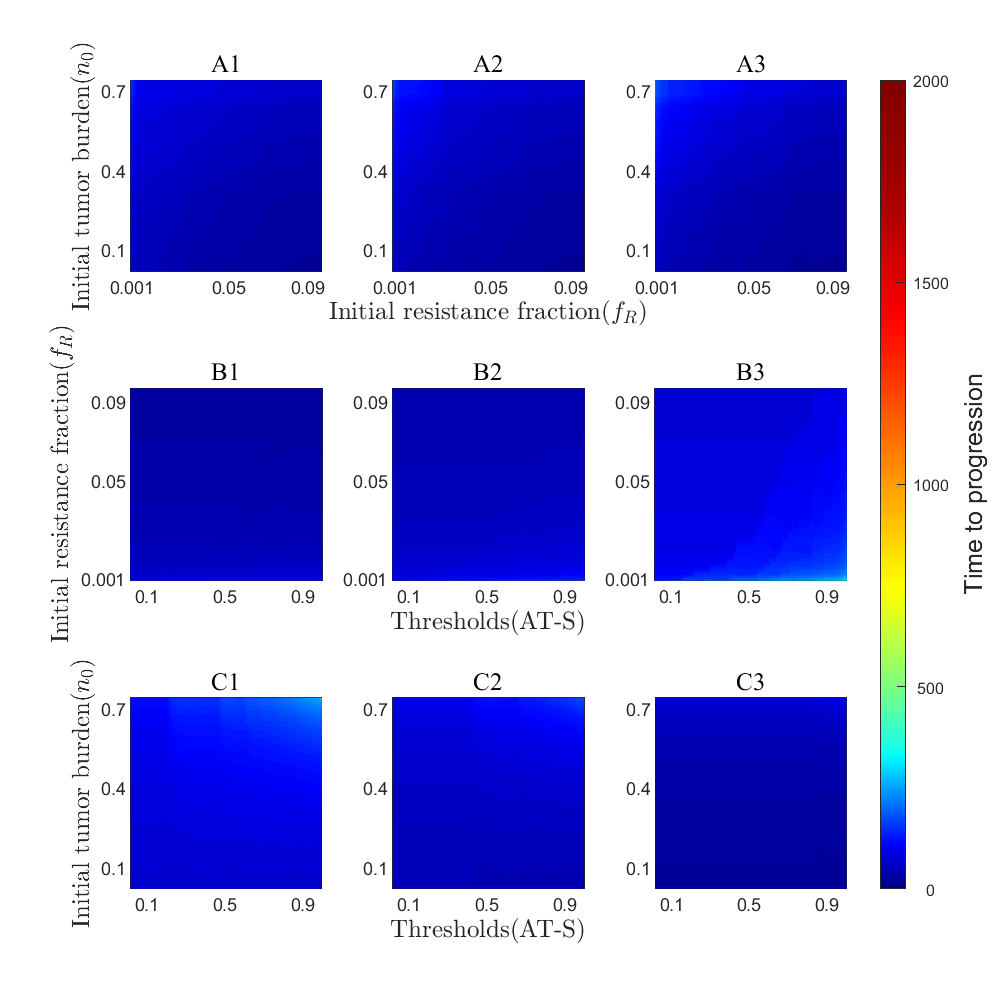}
\caption{The dependence of TTP on the thresholds, initial tumor burden ($n_0$) and initial fraction of resistant cells ($f_R$) in neutral competition with $r_S/r_R=0.8147$, ${r}_{S}=0.0915$.
\textbf{(A1)}-\textbf{(A3)} Effects of varying initial resistance proportions and cell counts on TTP, A1: $C_\mathrm{TH0}=0.3$; A2: $C_\mathrm{TH0}=0.5$; A3: $C_\mathrm{TH0}=0.8$
\textbf{(B1)}-\textbf{(B3)} Effects of varying treatment-holiday thresholds and initial resistance proportions on TTP. B1: $n_0=0.25$; B2: $n_0=0.5$; B3: $n_0=0.75$
\textbf{(C1)}-\textbf{(C3)} Effects of varying treatment-holiday thresholds and initial tumor burden on TTP, C1: $f_{R}=0.001$; C2: $f_{R}=0.01$; C3: $f_{R}=0.1$.
}
\label{fig:S:equal:supplement-3}
\end{figure}

\begin{figure}[tb]
\centering
\includegraphics[width=\textwidth]{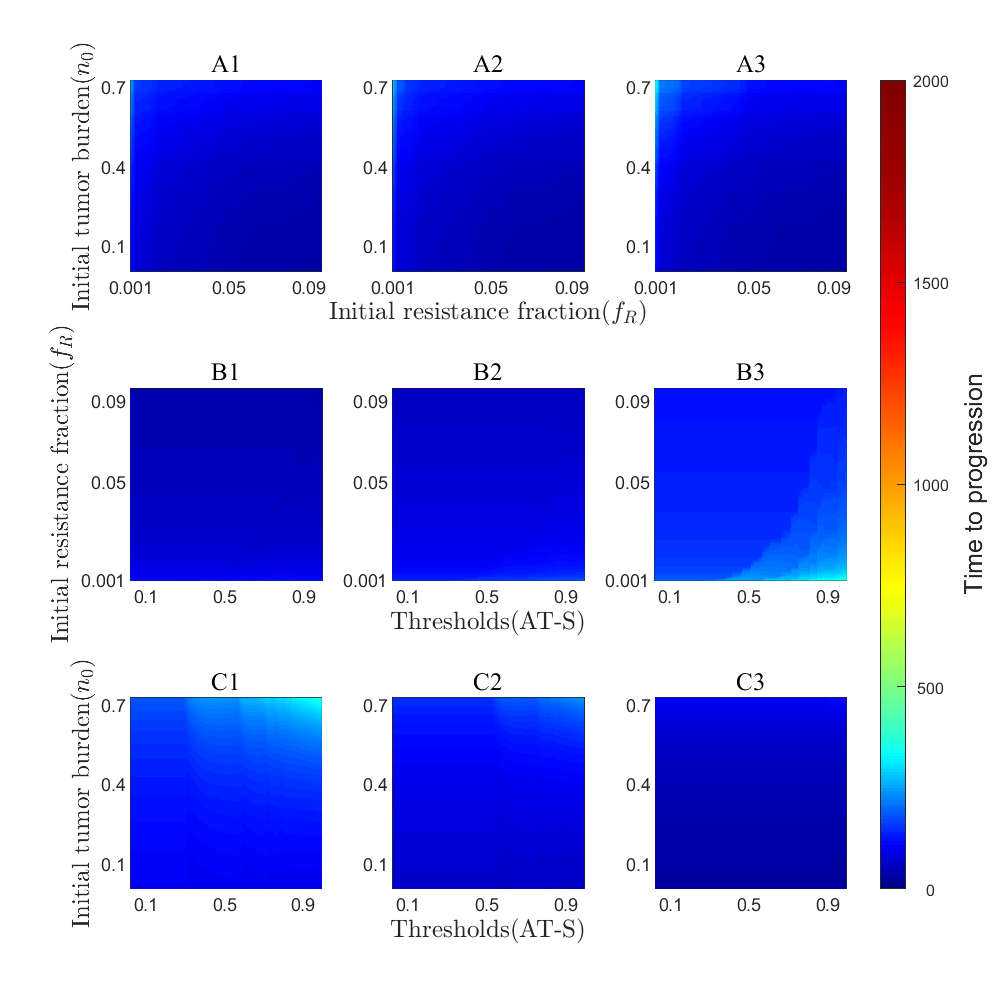}
\caption{The dependence of TTP on the thresholds, initial tumor burden ($n_0$) and initial fraction of resistant cells ($f_R$) in neutral competition with $r_S/r_R=0.6019$, ${r}_{S}=0.0445$.
\textbf{(A1)}-\textbf{(A3)} Effects of varying initial resistance proportions and cell counts on TTP, A1: $C_\mathrm{TH0}=0.3$; A2: $C_\mathrm{TH0}=0.5$; A3: $C_\mathrm{TH0}=0.8$
\textbf{(B1)}-\textbf{(B3)} Effects of varying treatment-holiday thresholds and initial resistance proportions on TTP. B1: $n_0=0.25$; B2: $n_0=0.5$; B3: $n_0=0.75$
\textbf{(C1)}-\textbf{(C3)} Effects of varying treatment-holiday thresholds and initial tumor burden on TTP, C1: $f_{R}=0.001$; C2: $f_{R}=0.01$; C3: $f_{R}=0.1$.
}
\label{fig:S:equal:supplement-4}
\end{figure}

\begin{figure}[tb]
\centering
\includegraphics[width=\textwidth]{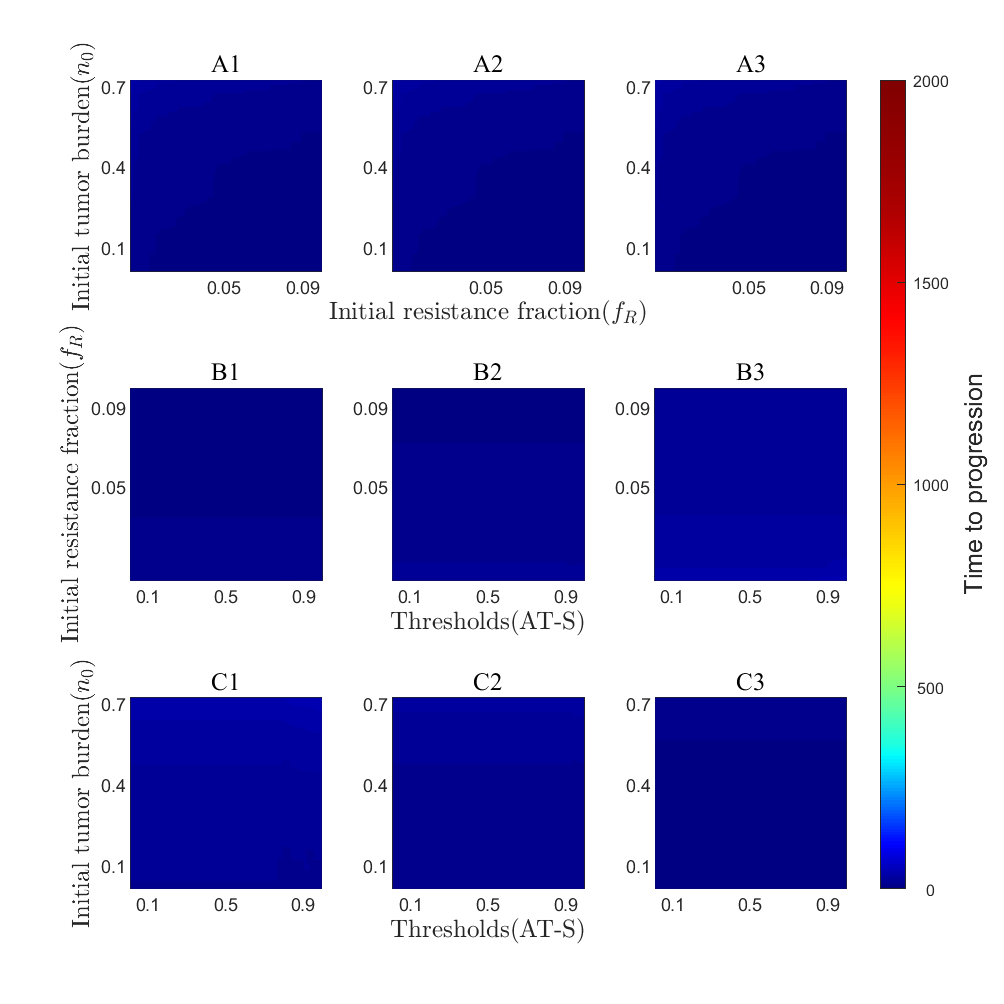}
\caption{The dependence of TTP on the thresholds, initial tumor burden ($n_0$) and initial fraction of resistant cells ($f_R$) in neutral competition with $r_S/r_R=0.1228$, ${r}_{S}=0.0695$.
\textbf{(A1)}-\textbf{(A3)} Effects of varying initial resistance proportions and cell counts on TTP, A1: $C_\mathrm{TH0}=0.3$; A2: $C_\mathrm{TH0}=0.5$; A3: $C_\mathrm{TH0}=0.8$
\textbf{(B1)}-\textbf{(B3)} Effects of varying treatment-holiday thresholds and initial resistance proportions on TTP. B1: $n_0=0.25$; B2: $n_0=0.5$; B3: $n_0=0.75$
\textbf{(C1)}-\textbf{(C3)} Effects of varying treatment-holiday thresholds and initial tumor burden on TTP, C1: $f_{R}=0.001$; C2: $f_{R}=0.01$; C3: $f_{R}=0.1$.
}
\label{fig:S:equal:supplement-5}
\end{figure}

\begin{figure}[tb]
\centering
\includegraphics[width=\textwidth]{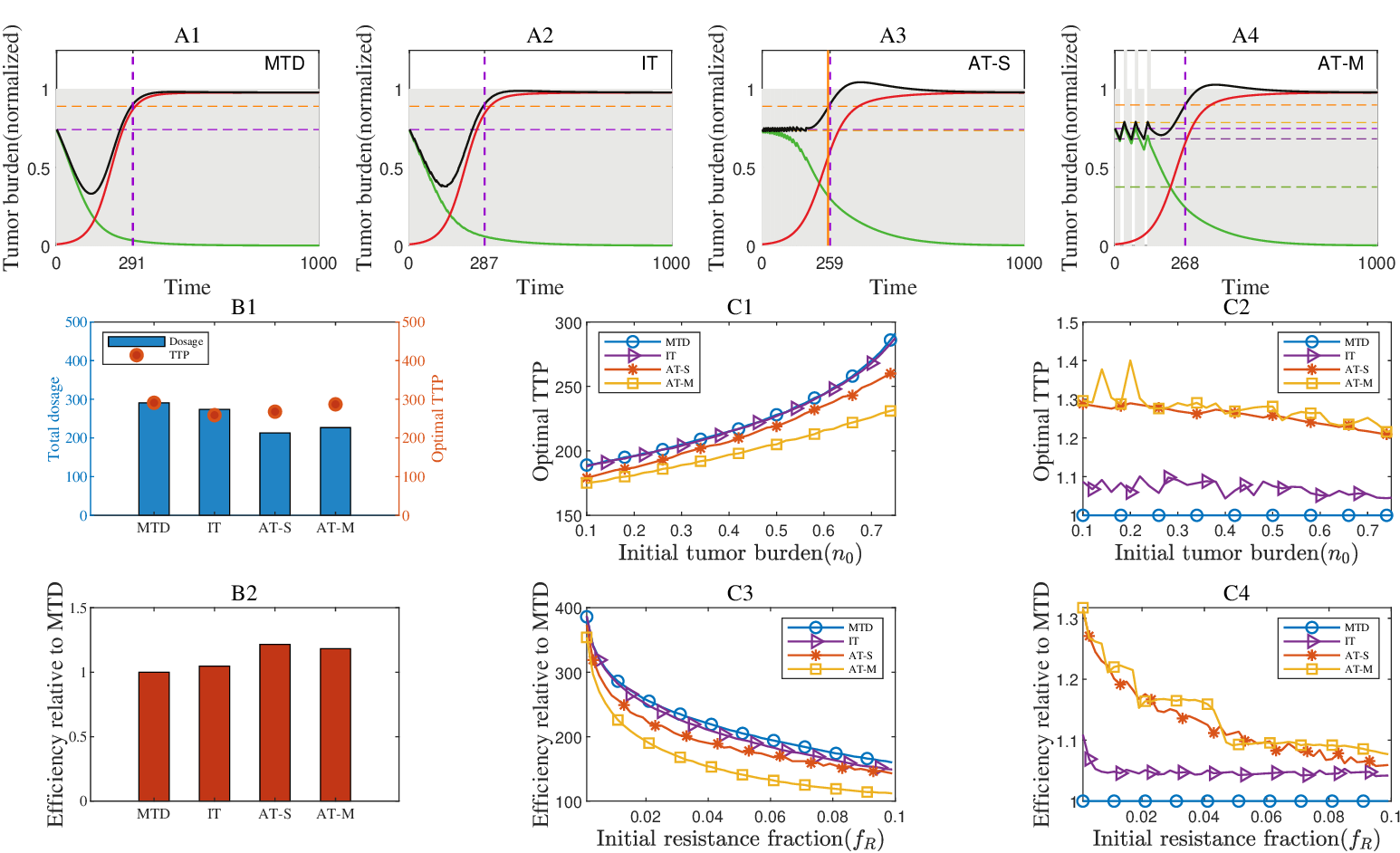}
\caption{Comparison of four treatment strategies on TTP in the scenario of \textbf{conditional-improve} where the parameters are taken as $\alpha=0.6$,$\beta=0.5$, $f_R=0.01$ and $n_0=0.75$.
\textbf{(A1)}-\textbf{(A4)} Tumor dynamics under four treatment strategies with optimal treatment parameters. Adaptive therapy outperforms MTD and IT. Optimal treatment parameters: $C_\mathrm{TH0}=0.98$, $\delta_1=0.25$, $\delta_2= 0.25$, $C_\mathrm{TH2}=1.05$ and $C_\mathrm{TH1}=1.07$, $T_D=9$, $T=10$.
\textbf{(B1)} Cumulative drug toxicity (total dosage) and corresponding TTP for the four strategies.
\textbf{(B2)} Relative efficiency for the four strategies.
\textbf{(C1)}-\textbf{(C2)} Effects of varying the initial tumor cell proportion on TTP and Relative efficiency of the four strategies.
\textbf{(C3)}-\textbf{(C4)} Effects of varying the initial resistant cell proportion on TTP and Relative efficiency of the four strategies.
}
\label{fig:S:leq-2}
\end{figure}

\begin{figure}[tb]
\centering
\includegraphics[width=\textwidth]{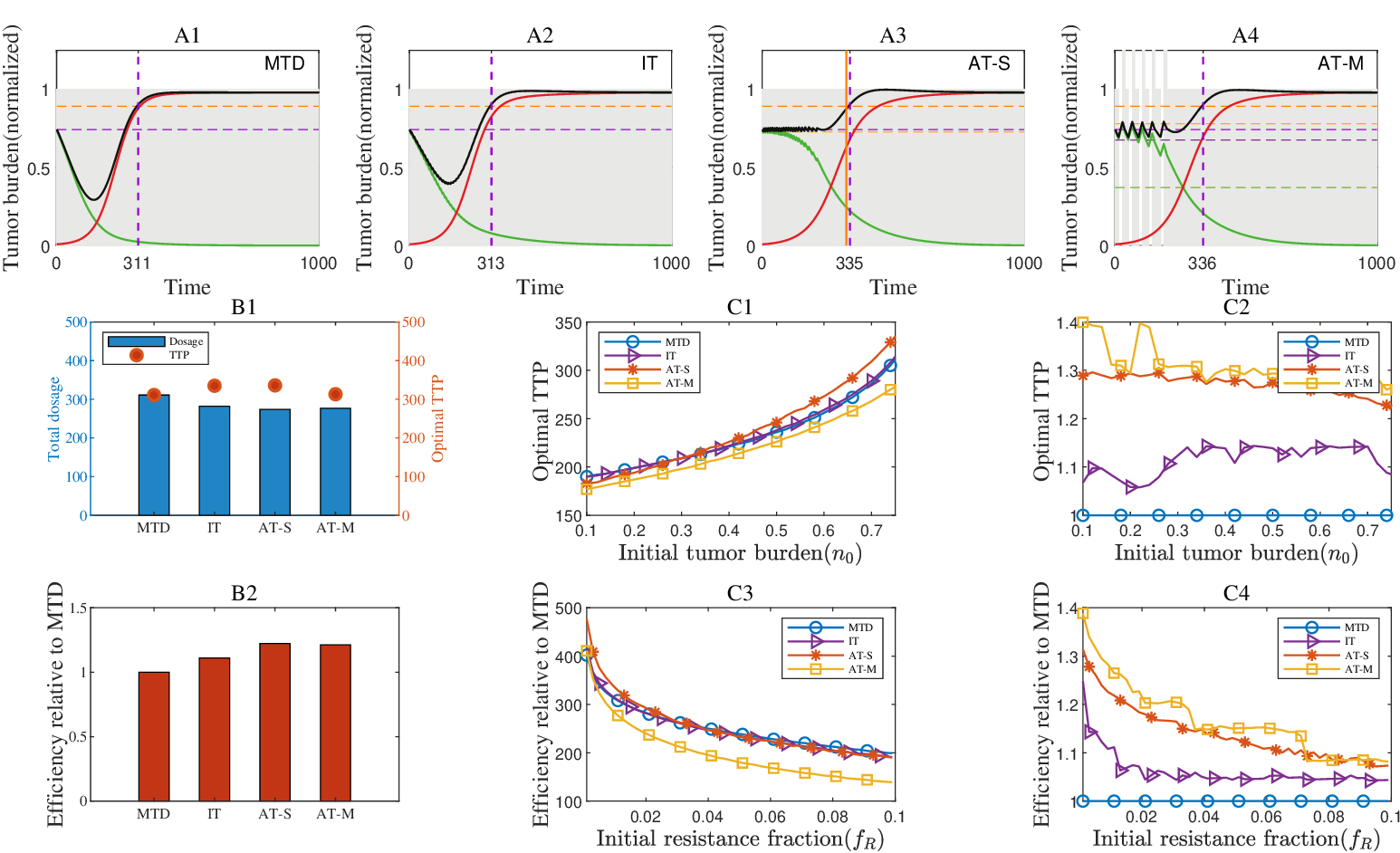}
\caption{Comparison of four treatment strategies on TTP in the scenario of \textbf{uniform-improve} scenario where the parameters are taken as $\alpha=0.6$, $\beta=0.3$, $f_R=0.01$ and $n_0=0.75$.
\textbf{(A1)}-\textbf{(A4)} Tumor dynamics of tumor evolution with four treatment strategies with optimal treatment parameters. Optimal parameters are taken as follows. $C_\mathrm{TH0}=0.98$, $\delta_1=0.25$, $\delta_2= 0.25$, $C_\mathrm{TH2}=1.05$ and $C_\mathrm{TH1}=1.09$, $T_D=20$, $T=21$.
\textbf{(B1)} Cumulative drug toxicity (total dosage) and corresponding TTP for the four strategies.
\textbf{(B2)} Relative efficiency for the four strategies.
\textbf{(C1)}-\textbf{(C2)} Effects of varying the initial tumor cell proportion on TTP and Relative efficiency of the four strategies.
\textbf{(C3)}-\textbf{(C4)} Effects of varying the initial resistant cell proportion on TTP and Relative efficiency of the four strategies.
}
\label{fig:S:leq-3}
\end{figure}

\begin{figure}[tb]
\centering
\includegraphics[width=\textwidth]{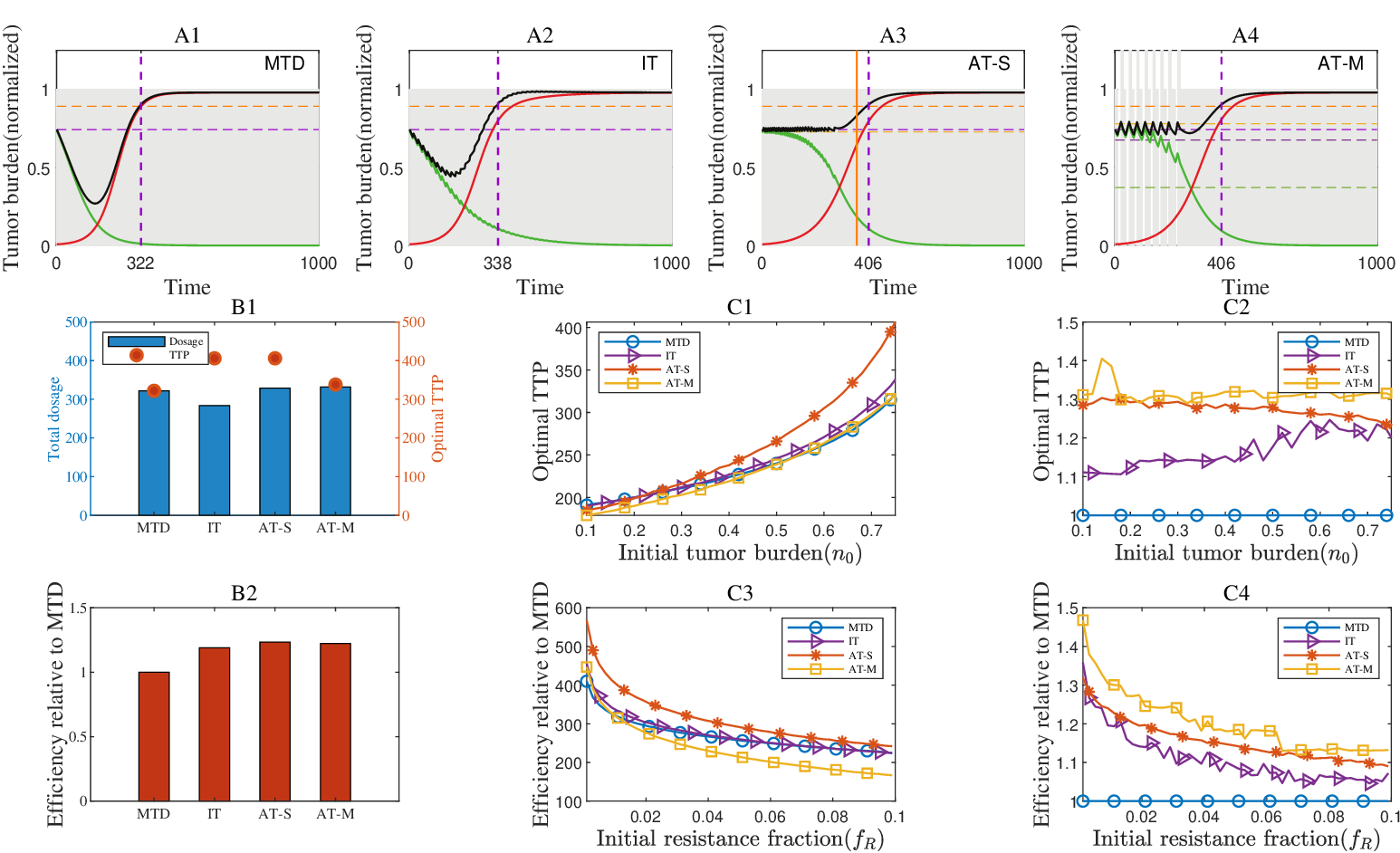}
\caption{Comparison of four treatment strategies on TTP in the scenario of \textbf{uniform-decline} where the parameters are taken as $\alpha=0.3$, $\beta=0.6$, $f_R=0.01$ and $n_0=0.75$)
\textbf{(A)} Tumor dynamics under four treatment strategies with optimal treatment parameters. Optimal parameters are as follows. $C_\mathrm{TH0}=0.98$, $\delta_1=0.25$, $\delta_2= 0.25$, $C_\mathrm{TH2}=1.05$ and $C_\mathrm{TH1}=1.05$, $T_D=15$, $T=18$.
\textbf{(B1)} Cumulative drug toxicity (total dosage) and corresponding TTP for the four strategies.
\textbf{(B2)} Relative efficiency for the four strategies.
\textbf{(C1)}-\textbf{(C2)} Effects of varying the initial tumor cell proportion on TTP and Relative efficiency of the four strategies.
\textbf{(C3)}-\textbf{(C4)} Effects of varying the initial resistant cell proportion on TTP and Relative efficiency of the four strategies.
}
\label{fig:S:leq-4}
\end{figure}

\begin{figure}[tb]
\centering
\includegraphics[width=\textwidth]{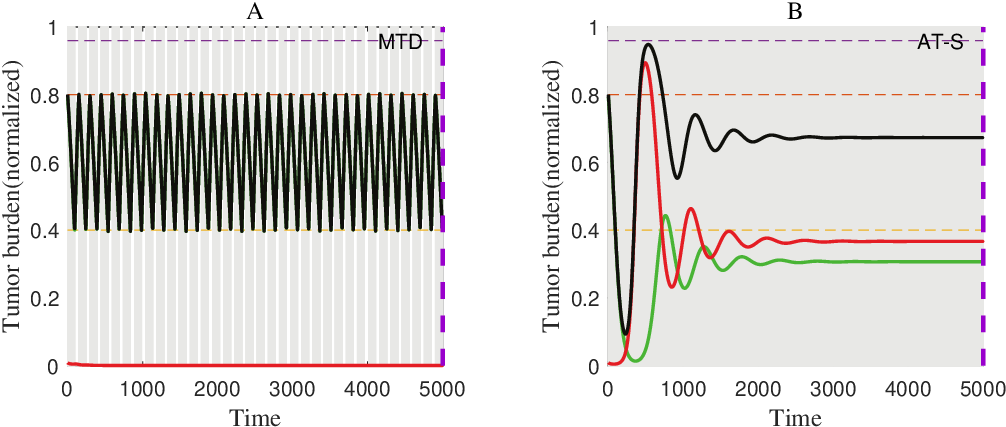}
\caption{Tumor evolution dynamics with $\alpha=2$, $\beta=2$. Purple dashed line: TTP; Red: Resistant cells; Green: Sensitive cells; Black: Total cell count; Shaded area: Treatment implementation.
\textbf{(A)} Adaptive therapy ($n_0=0.8$, $f_R=0.01$, $C_\mathrm{TH0}=0.5$).
\textbf{(B)} Maximum tolerated dose ($n_0=0.8$ , $f_R=0.01$).
}
\label{fig:S:supplement_geq_1}
\end{figure}

\end{document}